\documentclass[a4paper,UKenglish,cleveref,autoref,thm-restate]{lipics-v2021}
\usepackage[utf8]{inputenc}
\usepackage{multicol}
\usepackage[ruled,noend]{algorithm2e}
\usepackage[noend]{algorithmic} 
\usepackage{comment}
\usepackage{lineno}
\usepackage{tabularx}
\usepackage{makecell}
\usepackage{hyperref}
\usepackage{graphicx}
\usepackage{bm}
\usepackage{adjustbox}
\newtheorem{invariant}{Invariant}

\newcommand{\eps}{\varepsilon}

\newcommand{\poly}{\operatorname{poly}}
\title{Local Density and its Distributed Approximation}
\newcommand{\gee}{\textsl{g}}
\newcommand{\hee}{\textsl{h}}
\usepackage{float}
\usepackage{todonotes}
\hideLIPIcs
\nolinenumbers

\funding{This work was supported by the the VILLUM Foundation grant (VIL37507) ``Efficient Recomputations for Changeful Problems'', the Independent Research Fund Denmark grant 2020-2023 (9131-00044B) ``Dynamic Network Analysis'', and the European Union's Horizon 2020 research and innovation programme under the Marie Sk\l{}odowska-Curie grant agreement No 899987.
}

\author{Aleksander Bjørn Christiansen}{Technical University of Denmark, Kongens Lyngby, Denmark}{abgch@dtu.dk}{}{}
\author{Ivor van der Hoog}{Technical University of Denmark, Kongens Lyngby, Denmark}{idjva@dtu.dk}{https://orcid.org/0009-0006-2624-0231}{}
\author{Eva Rotenberg}{Technical University of Denmark, Kongens Lyngby, Denmark}{erot@dtu.dk}{https://orcid.org/0000-0001-5853-7909}{}

\authorrunning{Aleksander Bjørn Christiansen, Ivor van der Hoog, Eva Rotenberg}

\newtheorem{problem}{Problem}

\newcommand{\tor}{\!\to\!}
\begin{document}

\ccsdesc[500]{Theory of computation~Graph algorithms analysis}

\ccsdesc[500]{Theory of computation~Distributed algorithms}

\keywords{Distributed graph algorithms, graph density computation, graph density approximation, network analysis theory.}

\maketitle

\begin{abstract}
The densest subgraph problem is a classic problem in combinatorial optimisation.
Graphs with low maximum subgraph density are often called ``uniformly sparse'', leading to algorithms parameterised by this density. However, in reality, the sparsity of a graph is not necessarily uniform. This calls for a formally well-defined, fine-grained notion of density.

Danisch, Chan, and Sozio propose a definition for \emph{local density} that assigns to each vertex $v$ a value $\rho^*(v)$. This local density is a generalisation of the maximum subgraph density of a graph. I.e., if $\rho(G)$ is the subgraph density of a finite graph $G$, then $\rho(G)$ equals the maximum local density $\rho^*(v)$ over vertices $v$ in $G$. 
They present a Frank-Wolfe-based algorithm to approximate the local density of each vertex with no theoretical (asymptotic) guarantees. 

We provide an extensive study of this local density measure. 
Just as with (global) maximum subgraph density, we show that there is a dual relation between the local out-degrees and the minimum out-degree orientations of the graph.
We introduce the definition of the local out-degree $\gee^*(v)$ of a vertex $v$, and show it to be equal to the local density $\rho^*(v)$.
We consider the local out-degree to be conceptually simpler, shorter to define, and easier to compute. 

Using the local out-degree we show a previously unknown fact: that 
 existing algorithms already dynamically 
approximate 
the local density for each vertex with polylogarithmic update time. Next, we provide the first distributed algorithms that compute the local density with provable guarantees: given any $\varepsilon$ such that $\varepsilon^{-1} \in O(\operatorname{poly} n)$, we show a deterministic distributed algorithm in the LOCAL model where, after $O(\varepsilon^{-2} \log^2 n)$ rounds, every vertex $v$ outputs a $(1 + \varepsilon)$-approximation of their local density $\rho^*(v)$. 
In CONGEST, we show a deterministic distributed algorithm that requires $\text{poly}(\log n,\varepsilon^{-1}) \cdot 2^{O(\sqrt{\log n})}$ rounds, which is sublinear in $n$. 

As a corollary, we obtain the first deterministic algorithm running in a sublinear number of rounds for $(1+\varepsilon)$-approximate densest subgraph detection in the CONGEST model.
\end{abstract}

\setcounter{page}{0}

\thispagestyle{empty}

\newpage

\section{Introduction}

Density or sparsity measures of graphs are widely studied and have many applications. Examples include the arboricity, the degeneracy, and the maximum subgraph density, all of which are asymptotically related within a factor of $2$. Given a graph or subgraph $H$, its density, $\rho(H)$, is the average number of edges per vertex in $H$. The \emph{maximum subgraph density} $\rho^{\max}(G)$ of a graph $G$ is the maximum density $\rho(H)$ amongst all subgraphs $H\subseteq G$.

Computing maximum subgraph density has been studied both in the dynamic~\cite{chekuri2024adaptive,sawlani2020near}, streaming~\cite{BahmaniKV12,BhattacharyaHNT15} and distributed~\cite{SarmaLNT12,SuVu20} setting. 
Often these measures are used to parameterise the sparsity of ``uniformly sparse graphs''~\cite{Behnezhadetal,Brodal99dynamicrepresentations,Onaketal}. 
These measures are global measures in the sense that they measure the sparsity of the most dense part of the graph. 
In many cases the graph is not equally sparse (or dense) everywhere. 
Consider for example a lollipop graph: a large clique joined to a long path. 
The clique is a subgraph of high density, yet the vertices along the path sit in a part of the graph that is significantly less dense. 
Often, solutions for graph density related problems provide guarantees based on the most dense part of the graph.  
In some areas of computation more ``local'' solutions are desirable. Prior works of a global nature often completely disregard certain parts of the graphs, meaning that the output in sparser parts holds little to no information. We give three examples: 

$1)$ Many dynamic algorithms for estimating the subgraph density rely on modifying the solution locally. 
Algorithmic performance is expressed in terms of (global) graph sparsity, and thus fails to exploit the more fine-grained guarantee that local sparse areas yield.

$2)$ In network analysis, one is often interested in determining dense subgraphs as these subgraphs can be interpreted, for instance, as communities within a social network. 
However, since many classical algorithms are tuned towards only detecting the densest subgraphs, these algorithms might fail to detect communities in sparser parts of the network~\cite{Prac2,Prac1,Prac3}.

$3)$ Computing the maximum subgraph density is not very local, nor distributed. We consider the models LOCAL and CONGEST, and the lollipop graph. 
Here, almost instantly, the vertices of the clique realise they are part of a (very dense) clique. The vertices on the path may have to wait for diameter-many rounds before realising the maximum subgraph density of the graph.
Distributed algorithms that wish to compute the \emph{value} of the subgraph density are thus posed with a choice: either use $\Omega(D)$ rounds (where $D$ is the diameter of the graph), or let every vertex output a value that is at most the maximum subgraph density.

\subsection{Local density and results.}
We consider the definition of \emph{local density} $\rho^*(v)$ by Danisch, Chan, and Sozio~\cite{Prac2}, defined at each node $v$ of the graph (Definition~\ref{def:localdensity}).  Our contributions can be split into four categories, which we present in four sections with corresponding titles.

\subparagraph{A: Conceptual results for local density. }
Our primary contribution is an extensive overview of the theoretical properties of this local density measure.
We show that, just as in the maximum-subgraph density problem,  computing local density has a natural dual problem as we define the \emph{local out-degree}. 
Consider a (fractional) orientation of the graph that is \emph{locally fair}. i.e., for each directed edge $(u, v)$, the out-degree $g(u)$ is at most $g(v)$. 
We prove that for each vertex $v$, the out-degree of $v$ has the same value over all locally fair orientations. We define this value $\gee^*(v)$ as the local out-degree of $v$. 
We prove that the local density of each vertex is the dual of its local out-degree and thereby $\gee^*(v) = \rho^*(v)$. 
This new definition for local out-degree is considerably shorter than the definition for local density. It  allows us to show some previously unknown interesting properties of local density:

\subparagraph{B: Results for dynamic algorithms. }
We prove that in an \emph{approximately fair} orientation (a definition by Chekuri et al.~\cite{chekuri2024adaptive}) the out-degree $g(v)$ of each vertex is a $(1+\eps)$-approximation of $\gee^*(v) = \rho^*(v)$ (Theorem~\ref{thm:approx}). This implies a previously unknown fact: that there exist dynamic polylogarithmic algorithms~\cite{chekuri2024adaptive, christiansenICALP, ChristiansenNR23} where each vertex $v$ maintains a $(1+\eps)$-approximation of its local density $\rho^*(v)$ as by Danisch, Chan, and Sozio~\cite{Prac2}.

\subparagraph{C: Results in LOCAL. }
 We show that each node $v$ can obtain a $(1+\eps)$-approximation of $\rho^*(v)$ by surveying its $O(\eps^{-2}\log^2 n)$-hop neighbourhood. This induces a LOCAL algorithm where each vertex $v \in V$ computes a $(1 + \eps)$-approximation of $\rho^*(v)$ in $O(\eps^{-2} \log^2 n)$ rounds.

\emph{Commentary on runtime:}
We  observe that $\rho^{\max}(G)$ can be computed in LOCAL in $O(\eps^{-1} \log n)$ 
rounds.  In contrast, the stricter local density is computed in $O(\eps^{-2} \log^2 n)$ rounds. This gap may be explained by considering the local density $\rho^*(v)$ for low-local-density vertices $v$ in a graph that has high global density. 
The local density of $v$ can be affected by a dense subgraph within a hop distance of $\Theta(\eps^{-2} \log^2 n)$ (although it unclear if it can be affected \emph{enough} to prohibit a $(1+\eps)$-approximation of $\rho^*(v)$ in $o(\eps^{-2}\log^2 n)$ time).
The potential for a barrier of $O(\eps^{-2} \log^2 n)$ rounds is also illustrated by existing dynamic algorithms~\cite{chekuri2024adaptive, SuVu20} that maintain $\eta$-fair orientations. In this scenario, these algorithms have a worst-case recourse of $\Omega(\eps^{-2} \log^2 n)$.
We consider it an interesting open problem to either improve our running time in LOCAL, or, show that $O(\eps^{-2} \log^2 n)$ is tight.

 \subparagraph{D: Results in CONGEST}
We show a significantly more involved algorithm in CONGEST, where after $O( \poly \{ \eps^{-1}, \log n \} \cdot 2^{O(\sqrt{\log n})})$ rounds, 
each vertex $v$ outputs a $(1 + \eps)$-approximation of $\rho^*(v)$. 
Since $\max\limits_{v \in V} \rho^*(v) = \rho^{\max}(G)$, this is the first deterministic algorithm for $(1+\eps)$-approximating of the global subgraph density $\rho^{\max}(G)$ in CONGEST, that runs in a number of rounds that is sublinear in the diameter of the graph.

In the main body, we focus on the value variant where we want each vertex $v$ to output an approximation of $\rho^*(v)$. 
In Appendix~\ref{app:reporting} we extend our analysis so that each vertex $v$ can output a subgraph $H$ with $v \in H$ where $\rho(H)$ approximates $\rho^*(v)$. See also Tables~\ref{tab:results} and \ref{tab:results2}.

\begin{table}[h]
    \centering
       \begin{adjustbox}{max width=1.3\textwidth,center}
    \begin{tabular}{c|c|l|l|l}
        Model & Problem & Each $v$ outputs $\rho_v$ with & Rounds & Source \\ \hline
L & \ref{oldproblem:value}.1 & $\rho_v \in [(1+\eps)^{-1} \rho^{\max}(G), (1+\eps)\rho^{\max}(G) ]$ & $\Theta(D)$ & \cite{SuVu20} \\
 & \ref{oldproblem:value}.2 & $\max_{v} \rho_v \in [(1+\eps)^{-1} \rho^{\max}(G), (1+\eps)\rho^{\max}(G) ]$ & $O(\eps^{-1} \log n)$ & \cite{SuVu20} \\
 & \textbf{\ref{newproblem:value}} & $\pmb{\rho_v \in [(1+\eps)^{-1} \rho^*(v), (1+\eps)\rho^*(v) ]}$ & $\pmb{O(\eps^{-2} \log^2 n)}$ & \textbf{Cor.~\ref{cor:local}}      \\ \hline
C & \ref{oldproblem:value}.1 & $\rho_v \in [(2+\eps)^{-1} \rho^{\max}(G), (2+\eps)\rho^{\max}(G) ]$ & $O(D  \cdot \eps^{-1} \log n)$ & \cite{SarmaLNT12} \\ 
 & \ref{oldproblem:value}.1 & $\rho_v \in [(1+\eps)^{-1} \rho^{\max}(G), (1+\eps)\rho^{\max}(G) ]$ & \textcolor{orange}{ $O(\eps^{-4} \log^4 n + D)$ \hfill whp.}& \cite{SuVu20} \\
 & \ref{oldproblem:value}.2 & $\max_{v} \rho_v \in [(1+\eps)^{-1} \rho^{\max}(G), (1+\eps)\rho^{\max}(G) ]$ & \textcolor{orange}{ $O(\eps^{-4} \log^4 n)$ \hfill whp.}& \cite{SuVu20} \\
& \textbf{\ref{oldproblem:value}}\textbf{.2} &$\pmb{\rho_v \in [(1+\eps)^{-1} \rho^*(v), (1+\eps)\rho^*(v) ]}$ &$ \pmb{O(  \poly \{ \log n, \eps^{-1} \}) \cdot 2^{O(\sqrt{\log n})}}$  & \textbf{Thm.~\ref{thm:main}}   \\
& \textbf{\ref{newproblem:value}} &$\pmb{\rho_v \in [(1+\eps)^{-1} \rho^*(v), (1+\eps)\rho^*(v) ]}$ &$ \pmb{O(  \poly \{ \log n, \eps^{-1} \}) \cdot 2^{O(\sqrt{\log n})}}$  & \textbf{Thm.~\ref{thm:main}}   \\
  & \textbf{\ref{newproblem:value}} &   $\pmb{\rho_v \in [(1+\eps)^{-1} \rho^*(v), (1+\eps)\rho^*(v) ]}$ & \textcolor{orange}{$\pmb{O(\poly \{ \log n, \eps^{-1} \})}$ \hfill whp.} & \textbf{Thm.~\ref{thm:main}}    \\
    \end{tabular}
    \end{adjustbox}
\caption{Results in LOCAL (L) or CONGEST (C) where prior work for computing the global subgraph density is compared to our running time for the local subgraph density. D denotes the diameter. Orange running times are not deterministic and occur with high probability. }
    \label{tab:results}
\end{table}

\newpage
\section{Preliminaries and related work}
Let $G = (V, E)$ be an undirected weighted graph with $n$ vertices and $m$ edges. 
For any $v \in V$ and any integer $k$, we denote by $H_k(v)$ the $k$-hop neighborhood of $v$. 
For each edge $e \in E$ we denote by $\gee(e)$ the \emph{weight} of $e$. 
Any edge with endpoints $u,v$ may be denoted as $\overline{uv}$.

We can augment any weighted graph $G$ with a \emph{(fractional) orientation}. 
An orientation $\overrightarrow{G}$ assigns to each edge $\overline{uv}$ two positive real values: $\gee( u \tor v)$ and $\gee(v \tor u)$ such that $\gee(\overline{uv}) = \gee( u \tor v) + \gee( v \tor u)$. 
These values may be interpreted as pointing a fraction of the edge $\overline{uv}$ from $u$ to $v$, and the other fraction from $v$ to $u$. Given an orientation $\overrightarrow{G}$, we denote by $\gee(u) = \sum\limits_{v \in V} \gee(u \tor v)$ the \emph{out-degree} of $u$ (i.e., how much fractional edges point outwards from $u$ in $\overrightarrow{G}$).
Given these definitions, we can consider two graph measures of $G$: the maximum subgraph density and the minimum orientation of $G$.

\subparagraph{Global graph measures.}
For any subgraph $H \subseteq G$, its \emph{density} $\rho(H)$ is defined as $\rho(H) = \frac{1}{|V(H)|} \sum\limits_{e \in E(H)} \gee(e)$. 
The \emph{maximum subgraph density} $\rho^{\max}(G)$ is then the maximum over all $H \subseteq G$ of $\rho(H)$.  A subgraph $H \subseteq G$ is \emph{densest} whenever $\rho(H) = \rho^{\max}(G)$.  For any orientation $\overrightarrow{G}$ of $G$, its \emph{maximum out-degree} $\Delta(\overrightarrow{G})$ is the maximum over all $u$ of the out-degree $\gee(u)$. 
The \emph{optimal out-degree} of $G$, denoted by $\Delta^{\min}(G)$, is subsequently the minimum over all $\overrightarrow{G}$ of $\Delta(\overrightarrow{G})$. 
An orientation $\overrightarrow{G}$ itself is \emph{minimum} whenever $\Delta(\overrightarrow{G}) = \Delta^{\min}(G)$. 

The density of $G$ and the optimal out-degree are closely related. 
One way to illustrate this is through the following dual linear programs:
\vspace{-0.5cm}

\begin{align*}
 \hspace{-0.6cm}  \textbf{DS (Densest Subgraph)} &   & & & \Vert &  & \textbf{FO (Fractional Orientation)} & & \\
 \max \sum_{ \overline{uv} \in E} \gee(\overline{uv}) \cdot y_{u, v} & & \text{s.t.} &   &\Vert & & \min \rho  & & \text{s.t.} \\
x_u,x_v \geq y_{u,v}  & & \forall \overline{uv} \in E & &\Vert & & \gee(u \tor v) + \gee(v \tor u) \geq \gee(\overline{uv}) & &\forall \overline{uv}\in E \\
\sum_{v\in V} x_v  \leq 1 & &   & & \Vert & & \rho \geq \sum_{v\in V} \gee(u \tor v) & &\forall u\in V \\
\quad x_v,y_{u,v} \geq 0  & & \forall u, v \in V & & \Vert & & \gee(u \tor v), \gee(v \tor u) \geq 0 &  & \forall u, v \in V
\end{align*}

\noindent
Denote by $R$ the optimal value of DS and by $\Delta$ the optimal value of FO. By duality, $R = \Delta$. Moreover, Charikar~\cite{charikar2003greedy}  relates these two linear programs to the densest subgraph problem:

\begin{theorem}[Theorem 1 in \cite{charikar2003greedy}]
  Let $G$ be a unit weight graph.
Denote by $R$ the optimal solution of DS and by $D$ the optimal solution of $FO$. Then $\rho^{\max}(G) = R = \Delta = \Delta^{\min}(G)$.
\end{theorem}

\noindent
We show that this can be generalised to when $G$ is a weighted graph: 

\begin{restatable}{lemma}{equalmaxima}
    \label{lem:equalmaxima}
    Let $G$ be any weighted graph.
Denote by $R$ the optimal solution of DS and by $D$ the optimal solution of $FO$. Then $\rho^{\max}(G) = R = D = \Delta^{\min}(G)$.
\end{restatable}

\begin{proof}
    \textbf{We show that $R \geq \rho^{\max}(G)$.}
    First note that if $H \subseteq G$ is the densest subgraph of $G$, then for every edge $e \in E(H)$, $\hee(e) = \gee(e)$ (as decreasing weights only decreases the density of the subgraph). 
    There are finitely many subgraphs of $G$ with the same edge weights as $G$, and thus there exists a $H \subseteq G$ that is a densest subgraph of $G$.  
    We use $H$ to find a feasible solution to DS:

    For each $v \in V(H)$, we set $x_v = \frac{1}{|V(H)|}$.
    For every other $v \in V$, we set $x_v$ to be $0$. 
     For every $\overline{uv} \in E(H)$, we set $y_{u, v} = \frac{1}{|V(H)|}$.
     For every other $\overline{uv} \in E$, we set $y_{u, v} = 0$.
     This gives a feasible solution to DS. 
     Moreover: 
\begin{equation*}  
     \sum_{\overline{uv} \in E} \gee(\overline{uv}) \cdot y_{u, v} = \sum_{\overline{uv} \in E(H)} \gee(\overline{uv}) \cdot \frac{1}{|V(H)|} = \rho(H) = \rho^{\max}(G).
\end{equation*}

    \noindent
    It follows that $R \geq \rho^{\max}(G)$.

      \textbf{We show that $R \leq \rho^{\max}(G)$.}
Consider any solution to DS that realises the value $R$. 
Denote for a real value $r$ by $V(r) := \{ u \in V \mid x_u \geq r\}$.
Similarly, denote by $E(r) := \{ \overline{uv} \in E \mid y_{u, v} \geq r\}$.
Denote by $\theta[\cdot]$ the indicator function which gives $1$ when the expression $\cdot$ is true and $0$ otherwise.

\begin{equation}
    \label{eq:lessthan1}
\int_{0}^{\infty} |V(r) | \,dr = \int_{0}^{\infty} \sum_v \theta[ u_v \geq r ]\,dr = \sum_u \int_{0}^{x_u} 1 \, dr = \sum_u x_u \leq 1 
\end{equation}
Similarly, we can denote:

\begin{equation}
\label{eq:equalsr}
\int_{0}^{\infty} \sum_{\overline{uv} \in E(r)} \gee(\overline{uv}) \,dr = \int_{0}^{\infty} \sum_{\overline{uv}} \theta[ y_{u, v} \geq r ] \cdot \gee(\overline{uv})    \,dr = \sum_{\overline{uv}} \int_{0}^{y_{u, v}} \gee(\overline{uv}) \, dr = \sum_{\overline{uv}}  \gee(\overline{uv}) \cdot y_{u, v} = R 
\end{equation}

Now we claim that we can choose a value $\hat{r}$ such that $\frac{\sum_{\overline{uv} \in E(\hat{r})} \gee(\overline{uv})}{|V(\hat{r})|} \geq R$. Observe that the existence of $\hat{r}$ implies the lemma. 

Suppose for the sake of contradiction that such $\hat{r}$ does not exist.
Then per assumption:
\begin{equation*}
     \forall r \in (0, 1): \quad R \cdot |V(r)| \, \, - \sum_{\overline{uv} \in E(r)} \gee(\overline{uv}) > 0 \Rightarrow  \int_0^\infty  R \cdot |V(r)| - \sum_{\overline{uv} \in E(r)} \gee(\overline{uv}) \, dr > 0
\end{equation*}

However, this now gives a contradiction through Equations~\ref{eq:lessthan1} and \ref{eq:equalsr}:
\begin{align*}
    0 < \int_0^\infty  R \cdot |V(r)| \, \, - \sum_{\overline{uv} \in E(r)} \gee(\overline{uv}) \, dr = R \cdot \int_0^\infty |V(r)| dr - \int_0^\infty \sum_{\overline{uv} \in E(r)} \gee(\overline{uv}) \, dr \leq R - R = 0
\end{align*}

The proof that $D = \Delta^{\min}(G)$ is more straightforward:

Any orientation of $G$ is per definition is a solution to FO and so $\Delta^{\min}(G) \geq D$. 
What remains is to show that $\Delta^{\min}(G) \leq D$.
Suppose for the sake of contradiction that $D < \Delta^{\min}(G)$
and consider the solution to FO that realises the value $D$. 
Let this solution assign to every $\overline{uv} \in E$ weights $\gee(u \tor v) \geq 0$ and $\gee(v \tor u) \geq 0$. 
Because $D < \Delta^{\min}(G)$, it must be that the solution to $FO$ is not an orientation and so there must be edges $\overline{uv} \in E$ for which $\gee(u \tor v) + \gee(v \tor u) > \gee(\overline{uv})$.
Let $\gee(u \tor v) \geq \gee(v \tor u)$.
We may decrease $\gee(u \tor v)$ until $\gee(u \tor v) + \gee(v \tor u) = \gee(\overline{uv})$. 
By performing this action, we only reduce the sum $\sum_{v \in V} \gee(u \tor v)$ and thus we keep having a valid solution to FO. 
Doing this for every edge $\overline{uv} \in E$ creates an orientation where for all $u$, $\gee(u) = \sum_{v \in V} \gee(u \tor v) \leq R$.
Thus $\Delta^{\min}(G)\leq R$ --- a contradiction. 
\end{proof}

\subsection{Densest subgraph in dynamic algorithms}
In a classical, non-distributed model of computation  we can immediately formalise both the value variant of the (approximate) densest subgraph:

\begin{problem}
    \label{dynamic:value}
    Given a graph $G$ and an $\eps > 0$, output $\rho' \in [(1 + \eps)^{-1} \rho^{\max}(G), (1 + \eps) \rho^{\max}(G)]$.
\end{problem}

Alternatively, in the Fractional Orientation (FO) problem the goal is to output a $(1+\eps)$-approximation of $\Delta^{\min}(G)$. It turns out that FO is a more accessible problem to study. 
The LP formulations allow for a straightforward way to compute $\Delta^{\min}(G)$ and/or $\rho^{\max}(G)$.
However, solving the LP requires information about the entire graph, and this information is expensive to collect. 
Sawlani and Wang~\cite{sawlani2020near} get around this difficulty by instead solving an approximate version of (FO). They work with a concept we call \emph{local fairness}.

\begin{definition}
Let $\overrightarrow{G}$ be a fractional orientation of a graph $G$.
We say that $\overrightarrow{G}$ is locally fair whenever  $\gee(u \tor v) > 0$ implies $\gee(u) \leq \gee(v)$. 
\end{definition}

\noindent
Chekuri et al.~\cite{chekuri2024adaptive} extend this definition to $\eta$-fairness:  

\begin{definition}
Let $\overrightarrow{G}$ be a fractional orientation of a graph $G$.
We say that $\overrightarrow{G}$ is $\eta$-fair (for $\eta > 0$) whenever $\gee(u \tor v) > 0$ implies that $\gee(u) \leq (1 + \eta) \gee(v)$. 
\end{definition}

\subparagraph{Related work in dynamic algorithms} Chekuri et al.~\cite{chekuri2024adaptive} continue to focus on computing a $(1 + \eps)$-approximation of the Densest Subgraph problem.
They show that, if $G$ is a unit weight graph, there exists a $(1 + \eps)$-approximate solution to FO that is $\eta$-fair (for some smartly chosen $\eta$). 
They subsequently prove that an $\eta$-fair orientation allows you to find a $(1 + \eps)$-approximate densest subgraph. 
This allows them to dynamically maintain the value of the  densest subgraph of $G$ in $O( \eps^{-6} \log^3 n \log \rho^{\max}(G))$ time per insertion or deletion of edges in $G$. 
By leveraging the $\eta$-fairness of the orientation, they can report a $(1 + \eps)$-approximate densest subgraph in time proportional to the size of the subgraph. 

\subsection{Approximate densest subgraph in LOCAL and CONGEST}
\label{sub:suvu}

We focus on the \emph{value} variant of the problem, where each vertex outputs a value (as opposed to the \emph{reporting} variant in Appendix~\ref{app:reporting}, where the goal is to report a densest subgraph).

 \begin{problem}
    \label{oldproblem:value}
     Given a graph $G$ and an $\eps > 0$, each vertex $v$ outputs a value $\rho_v$ and either:
     \begin{itemize}
      \item    \textbf{Problem 2.1:} we require that  $\, \forall v, \rho_v \in [(1 + \eps)^{-1} \rho^{\max}(G), (1 + \eps) \rho^{\max}(G)]$, or
             \item  \textbf{Problem 2.2}: we require that $\,  \max_v \rho_v \in [(1+\eps)^{-1} \rho^{\max}(G), (1+\eps)\rho^{\max}(G) ]$.
     \end{itemize}
 \end{problem}

\subparagraph{Related work. }
Problem \ref{oldproblem:value}.1 has a trivial $\Omega(D)$ lower bound, obtained by constructing a lollipop graph (where $D$ denotes the diameter). 
In LOCAL, it is trivial to solve Problem \ref{oldproblem:value}.1 in $\Theta(D)$ time. 
Problem \ref{oldproblem:value}.2 was studied by Ghaffari and Su~\cite{GhaffariS17} who present a randomised $(1 + \eps)$-approximation in LOCAL that uses $O(\eps^{-3} \log^4 n)$ rounds. 
Fischer et al.~\cite{fischer2017deterministic} present a deterministic $(1 + \eps)$-approximation in LOCAL  that uses $2^{O(\log^2 ( \eps^{-1} \log n))}$ rounds. 
Ghaffari et al.~\cite{ghaffari2018derandomizing} improve this to $O(\eps^{-9} \log^{15}  n )$ rounds.
The work by Harris~\cite{harris2019distributed} improves this to $\tilde{O}(\eps^{-6} \log^6 n)$ rounds. 
Su and Vu~\cite{SuVu20} present the state-of-the-art in this area. 
They prove that for any graph $G$, 
 there exists a vertex $v$ such that for the $k$-hop neighbourhood $H_k(v)$ (with $k  \in O(\eps^{-1} \log n)$) the density $\rho^{\max}(H_k)$ is a $(1 + \eps)$-approximation of $\rho^{\max}(G)$. 
This immediately leads to a trivial LOCAL algorithm: each vertex $u$ collects its $k$-hop neighbourhood $H_k(u)$ in $O(\eps^{-1} \log n)$ rounds, solves the LP of Densest Subgraph in its own node, and reports the value $\rho^{\max}(H_k(u))$. 

In CONGEST, the state-of-the-art deterministic algorithm for Problem~\ref{oldproblem:value}.1 and \ref{oldproblem:value}.2 is by Das Sarma et al.~\cite{SarmaLNT12} who present a $(2+\eps)$-approximation in $O(D \cdot \eps^{-1} \log n)$ rounds.
The best randomised work is by Su and Vu~\cite{SuVu20} who present a randomised algorithm for Problem~\ref{oldproblem:value}.2 that runs in $O(\eps^{-4} \log^4 n)$ rounds w.h.p. See also Table~\ref{tab:results} for an overview.

\subsection{Local density}
Danisch, Chan, and Sozio~\cite{Prac2} introduce a more local measure which they call the \emph{local density}. 
Its lengthy definition assigns to each vertex $v$ a value.
We note for the reader that we almost immediately define our local out-degree (Definition~\ref{def:local_out_degree}), and only use local out-degree in proofs. Hence, the reader is not required to have a thorough understanding of the following:

\begin{definition}[Definition 2.2 in ~\cite{Prac2}]
    Let $G = (V, E)$ be a weighted graph where an edge $e$ has weight $\gee(e)$. 
    Let $B \subseteq V$.
    For any $X \subseteq V - B$, we define the \emph{quotient edges} $\hat{E}_B(X)$ as all edges in $G$ with one endpoint in $X$, and the other endpoint in $X$ or $B$.  We define:  
    \begin{itemize}
        \item for $X \subseteq V - B$, the \emph{quotient subgraph density} $\hat{\rho}_B(X) :=  \frac{1}{|X|} \sum\limits_{e \in \hat{E}_B(X)} \gee(e)$.
        \item the \emph{maximum quotient density} $\hat{\rho}_B(G) := \max\limits_{X \subseteq V - B } \, \hat{\rho}_B(X)$. 
    \end{itemize} 
\end{definition}

\begin{definition}[Definition 2.3 in ~\cite{Prac2}]
\label{def:diminishing}
Given a weighted undirected graph $G = (V, E)$, we define the diminishing-dense decomposition $\mathcal{B}$ of $G$ as the sequence $B_0 \subset B_1 \ldots \subset B_\ell = V$:

We define $B_0  = \emptyset$. 
For $i \geq 1$ if $B_{i-1} =V$ then $\ell := i$. 
Otherwise:
\[
S_i := arg  \, \max_{X \subseteq V -  B_{i-1}} \hat{\rho}_{B_{i-1}}(X), \textnormal{ and }  B_i := B_{i-1} \cup S_i.
\]
\end{definition}

\begin{definition}[Definition 2.3 in~\cite{Prac2}]
    \label{def:localdensity}
    Given a weighted undirected graph $G = (V, E)$ and a diminishing-dense decomposition $\mathcal{B}$, each vertex $v \in V$ has one integer $i$ where $v \in S_i$. We define the \emph{local density} $\rho^*(v) := \hat{\rho}_{B_{i-1}}(S_i)$.  
\end{definition}

\subparagraph{The benefit of local measures:}
Problem~\ref{oldproblem:value}'s variants have drawbacks in a distributed model of computation. Problem~\ref{oldproblem:value}.1 has an $\Omega(d)$ lower bound (making it trivial in LOCAL). Problem~\ref{oldproblem:value}.2 allows some vertices to output nonsense. 
The definition of local density alleviates these issues, as we may define an algorithmic problem which we consider to be more natural: 

\begin{problem}
    \label{newproblem:value}
    Given $(G, \eps)$, each vertex $v$ outputs $\rho_v \in [(1 + \eps)^{-1} \rho^*(v) , (1 + \eps) \rho^*(v) ]$. 
\end{problem}

\section{Results and organisation}

Now we are ready to formally state our contributions.
Our primary contribution is that we show a dual definition to local density, which we call the local out-degree:

\begin{definition}\label{def:local_out_degree}
    Given a graph $G = (V, E)$, we define the \emph{local out-degree} as:
    \[
 \gee^*(u) := \textnormal{ the out-degree } \gee(u) \textnormal{ in any locally fair fractional orientation of } G.
\]
\end{definition}

\noindent
It is not immediately clear that the local out-degree is well-defined. We prove (Theorem~\ref{thm:localDensity}) that each vertex in $G$ has the same out-degree across all locally fair orientations of $G$ (and thus, the set of all locally fair orientations of $G$ assigns to each vertex a real value). 
We believe that the local out-degree is conceptually simpler that the local density. Through this definition, we are able to show various algorithms to approximate the local density.

\renewcommand\thesubsubsection{\thesection.\Alph{subsubsection}}
\subsubsection{Conceptual results for local density }
We prove in Section~\ref{sec:definition} that these local definitions generalise the global definition of subgraph density and out-degree, as they exhibit the same dual behaviour. We show several previously unknown properties of the local density,  which we consider to be of independent interest:

\begin{restatable}{theorem}{fairImpliesDense}
\label{thm:localDensity}
    For any weighted graph $G$, $\forall v \in V$, $\gee^*(v)$ is well-defined and equals $\rho^*(v)$.
\end{restatable}

\begin{restatable}{corollary}{equal}
\label{cor:equal}
Given a weighted graph $G$, $\rho^{\max}(G)= \Delta^{\min}(G) = \max_v \gee^*(v)$.
\end{restatable}

\begin{restatable}{corollary}{alwaysFair}
\label{cor:alwaysfair}
    For any graph $G$, there exists a fractional orientation $\overrightarrow{G}$ that is locally fair. 
\end{restatable}

\subsubsection{Results for dynamic algorithms}
We show in Section~\ref{sec:etafair} that $\eta$-fair orientations imply approximations for our local measures:

\begin{restatable}{theorem}{approx}
    \label{thm:approx}
    Let $G$ be a weighted graph and $\overrightarrow{G}$ be an $\eta$-fair fractional orientation for $\eta \leq \frac{\eps^2}{128 \cdot \log n}$. Then $\forall v \in V$:  $
     (1 + \eps)^{-1}\rho^*(v) \leq \gee(v) \leq (1 + \eps)\rho^*(v).
    $   
\end{restatable}                           

\noindent
This immediately implies the following Corollary by applying~\cite{chekuri2024adaptive}:
\begin{restatable}{corollary}{dynamic}
    \label{cor:dynamic}
    There exists an algorithm~\cite{chekuri2024adaptive} that can fractionally orient a dynamic unit-weight graph  $G$ with $n$ vertices subject to edge insertions and deletions with  deterministic worst-case $O(\eps^{-6} \log^4 n))$ update time such that for all  $v \in V$:
    \[ \gee(v) \in [(1 + \eps)^{-1} \rho^*(v), (1 + \eps) \rho^*(v)]. \]
\end{restatable}

\subsubsection{Results in LOCAL}
The local density as a measure is not entirely local. 
However, we prove in Section~\ref{sec:local} that far-away subgraphs affect the local density of a vertex $v$ only marginally:

\begin{restatable}{theorem}{newlocal}
\label{thm:newlocal}
    Let $G$ be a unit-weight graph. 
    For any $\eps > 0$ and vertex $v$, denote by $\rho^*(v)$ its local density and by $\rho^*_k(v)$ its local density in $H_k(v)$. Then $\rho^*_k(v) \in [ (1 + \eps)^{-1} \rho^*(v), (1 + \eps) \rho^*(v)]$ for $k \in \Theta( \eps^{-2} \log^2 n)$. \end{restatable}

\noindent
This immediately implies a trivial algorithm for problem~\ref{oldproblem:value} in LOCAL (where each vertex $v$ collects its $k$-hop neighbourhood $H_k(v)$ for $k \in \Theta(\eps^{-2} \log^2 n)$ and then solves FO on $H_k(v)$): 

\begin{restatable}{corollary}{local}
    \label{cor:local}
    There exists an algorithm in LOCAL that given a unit graph $G$ and $\eps > 0$ computes in $O(\eps^{-2} \log^2 n)$ rounds for all  $v \in V$ a value $\rho_v \in [(1 + \eps)^{-1} \rho^*(v), (1 + \eps) \rho^*(v)]$.
\end{restatable}

\subsubsection{Results in CONGEST}
Finally in Section~\ref{sec:comp}, we solve Problem~\ref{newproblem:value} in CONGEST by computing an $\eta$-fair orientation. We use as a subroutine algorithm to compute \emph{blocking flows} in an $h$-layered DAG~\cite{haeupler2023maximum}:

\begin{restatable}{theorem}{main}
\label{thm:main}
      Suppose one can compute a blocking flow in an $n$-node $h$-layered DAG
in $\textnormal{Blocking}(h, n)$ rounds. 
    There exists an algorithm in CONGEST that given a unit-weight graph $G$ and $\eps > 0$ computes in $O(  \eps^{-3} \log^4 n \cdot (\eps^{-2} \log^2 n + \textnormal{Blocking}(\eps^{-2} \log^2 n, n))  )$ rounds an orientation 
    $\overrightarrow{G}$ such that for all $ v \in V$: $\gee(v) \in [(1 + \eps)^{-1} \rho^*(v)), (1 + \eps)\rho^*(v) \rho^*(v)]$.%
\end{restatable}

As a corollary,  we obtain the first deterministic algorithm running in a sublinear number of rounds for $(1+\varepsilon)$-approximate dense subgraph detection in the CONGEST model (Table~\ref{tab:results}).



\section{Conceptual results for local density}
\label{sec:definition}

Our primary contribution is  the definition of local out-degree as a dual to local density.

\begin{lemma}
    \label{lem:welldefined}
    For any two locally fair orientations $\overrightarrow{G}$ or $\overrightarrow{G}'$ where a vertex $u$ has out-degree $\gee(u)$ or $\gee'(u)$ respectively, $\gee(u) = \gee'(u)$.
\end{lemma}

\begin{proof}
Suppose for the sake of contradiction that there exists two locally fair orientations  $(\overrightarrow{G}, \overrightarrow{G}')$ and a vertex $u \in V$ where $\gee(u) > \gee'(u)$.
We define their \emph{symmetric difference} graph $S$ as a digraph where the vertices are $V$ and there exists an edge $\overrightarrow{ab}$ whenever $\gee(a \tor b) > \gee'(a \tor b)$.
We may assume that $S$ contains no directed cycles:

Indeed if $S$ contains any directed cycle $\pi$ we change $\overrightarrow{G'}$, where for all $\overrightarrow{ab} \in \pi$ we slightly increase $\gee'(a \tor b)$ until $S$ loses an edge. 
This operation does not change the out-degree of any vertex in $\overrightarrow{G}'$. So, we still have two locally fair orientations $(\overrightarrow{G}, \overrightarrow{G}')$  with $\gee(u) > \gee'(u)$. 

Since  $\gee(u) > \gee'(u)$, the vertex $u$ must have at least one out-edge in $S$ and since $S$ has no cycles, it follows that $u$ must have some directed path $\pi_v$ to a sink $v$ in $S$. 
Since $v$ is a sink in the symmetric difference graph it follows that $\gee(v) < \gee'(v)$. 

\noindent
 However we now observe the following property of the path $\pi_v$:
\begin{itemize}
    \item $\forall \overrightarrow{ab} \in  \pi_v$, $\gee(a \tor b) > \gee'(a \tor b)$. Thus, $\gee(a \tor b) > 0$ and so there exists a directed path from $u$ to $v$ in $\overrightarrow{G}$. Since    
    $\overrightarrow{G}$ is locally fair this implies that $\gee(u) \leq \gee(v)$.
    \item $\forall \overrightarrow{ab} \in  \pi_v$, $\gee(a \tor b) > \gee'(a \tor b)$. Thus, $\gee'(a \tor b) < \gee(ab)$ and so $\gee'(b \tor a) > 0$.  It follows that there exists a directed path from $v$ to $u$ in $\overrightarrow{G}'$. Local fairness implies $\gee'(v) \leq \gee'(u)$.
\end{itemize}
The 4 equations: $\gee(u) > \gee'(u)$, $\, \gee(v) < \gee'(v)$, $\, \gee(u) \leq \gee(v)$, and $\, \gee'(u) \geq \gee'(v)$ give a contradiction. 
\end{proof}

Lemma~\ref{lem:welldefined} would imply that the local out-degree is well-defined, \emph{if} the set of locally fair orientations is non-empty. 
Bera, Bhattacharya, Choudhari and Ghosh already aim to prove this in Section 4.1 of~\cite{BeraBCG22} (right above Equation 9).  
They claim that a locally fair orientation always exist by the following argument:
They consider an arbitrary orientation $\overrightarrow{G}$ that is not locally fair.
They claim that for any pair $(u, v)$ where $\gee(u) > \gee(v)$ and $\gee(u \tor v) > 0$ it is possible to transfer some out-degree from $\gee(u)$ to $\gee(v)$.
The existence of a locally fair orientation would follow, if it can be shown that this procedure converges to a locally fair orientation. Indeed, since the space of all orientations is a compact polytope, the limit of a converging sequence over this domain must lie within the space. 

However, it is intuitive but not clear that this procedure indeed converges. Indeed, decreasing $\gee(u)$ and increasing $\gee(v)$ may cause some other edge $(w, u)$ or $(v, w)$ to start violating local fairness. One way to show convergence is to define a potential function and to show that such a transfer always decreases the potential. 
We define the potential function $\sum_{v \in V} \gee^2(v)$, thereby creating a quadratic program where the domain is the space of all orientations of $G$. 
We prove that any optimal solution to this quadratic program must be a locally fair orientation. Any quadratic function over a compact domain has an optimum and so the existence of a locally fair orientation follows.

\fairImpliesDense*

\begin{proof}
     We consider the following quadratic program $\textbf{FO}^2$ from \cite[Section 4]{Prac2} where we compute a fractional orientation of the graph $G$ subject to a quadratic cost function: 
\begin{align*}
  \min \sum \gee(u)^2  & & \text{s.t.} \\
   \gee(u \tor v) + \gee(v \tor u) \geq \gee(\overline{uv}) & &\forall \overline{uv}\in E \\
  \gee(u) \geq \sum_{v\in V} \gee(u \tor v) & &\forall u\in V \\
  \gee(u \tor v), \gee(v \tor u) \geq 0 & & \forall u, v \in V
\end{align*}

Consider any optimal solution to the quadratic program. It must be that $\gee(u) = \sum_{v \in V} \gee(u \tor v)$. 
Danisch, Chan, and Sozio~\cite[Corollary 4.4]{Prac2} prove that for any vertex $u$, the local density $\rho^*(u) = \gee(u)$.

We first note that any solution to the quadratic program is an orientation. Indeed, suppose for the sake of contradiction that there exists an edge $\overline{uv} \in E$ where $\gee(u \tor v) + \gee(v \tor u) > \gee(\overline{uv})$.
We may now decrease either $\gee(u \tor v)$ or $\gee(v \tor u)$ to obtain another viable solution to the program. 
Consider decreasing $\gee(u \tor v)$, then we may decrease $\gee(u)$ and maintain a viable and better solution to the program -- a contradiction.

Secondly, we claim that the optimal solution to the quadratic program is a locally fair orientation. Suppose for the sake of contradiction that there exist $u, v$ with $\gee(u) = \gee(v) + \delta'$ and $\gee(u \tor v) = \delta$ for $\delta, \delta' > 0$. 
We can decrease $\gee(u \tor v)$ to zero by increasing $\gee(v \tor u)$ by $\Delta = \min\{ \delta, \delta' \}$ and still maintain a solution to the program.
This reduces the solution's value by $ (\gee(u) - \Delta)^2 - (\gee(v) + \Delta)^2$. 
However, we now found a solution to the quadratic program with a lower value than the optimal solution -- a contradiction.

Thus, the solution to $\textbf{FO}^2$ gives a locally fair orientation $\overrightarrow{G}$ where each vertex $u$ has out-degree $\gee(u)$.

The local density $\gee^*(u) = \gee(u)$ is by Lemma~\ref{lem:welldefined} well-defined. Danisch, Chan and Sozio~\cite[Corollary 4.4]{Prac2} show that $\rho^*(u) = \gee(u)$, which proves the theorem. 
\end{proof}

\noindent
Since the local density equals the local out-degree, we conclude from~\cite{Prac2} that:

\equal*

Since a quadratic program over a convex domain always has a solution, we may also note the following interesting fact:

\alwaysFair*

\section{Results for dynamic algorithms}
\label{sec:etafair}

We use our definition of local out-degree to show that there already exist dynamic algorithms that approximate the local density of each vertex. 
Recall that an orientation $\overrightarrow{G}$ is $\eta$-fair whenever for all $\overline{uv} \in E(\overrightarrow{G})$, $\gee(u \tor v) > 0$ implies that $\gee(u) \leq (1 + \eta) \gee(v)$. 
We show that if we choose $\eta \leq \frac{\eps^2}{128 \cdot \log n}$, then for  any $\eta$-fair orientation, for all $v$, the out-degree $\gee(v)$ is a $(1 + \eps)$ approximation of $\gee^*(v) = \rho^*(v)$. 
Moreover, we prove that the maximal local out-degree (i.e. $\max_{u \in V} \gee(u)$) is a $(1 + \eps)$ approximation of $\Delta^{\min}(G) = \rho^{\max}(G)$; illustrating that approximating the local measures is a strictly more general problem. 

\begin{restatable}{lemma}{estimate}
\label{lem:estimate}
Let $\eta \leq \frac{\eps^2}{128 \cdot \log n}$ and $k \leq \log_{(1 + \frac{1}{16} \eps)} n$.
Then $(1 + \eta)^{-k} \geq (1 + 0.5 \eps)^{-1}$. \end{restatable}

\begin{proof}
    Using $\log(1 + x) \geq x / 2$ whenever $x < 1$, we obtain:
\[
-\log_{1+\frac{\eps}{16}}(n) = -\frac{\log(n)} {\log 1+\frac{\eps}{16}}  \geq -\frac{\log(n)} {\frac{\eps}{32}} \geq -\frac{\eps}{4} \cdot{}\frac{128\log(n)} {\eps^2}
\]

\begin{align*}
     (1+\eta)^{-k} \geq  (1+\frac{\eps^2}{128 \cdot c \cdot \log n})^{- c \cdot \log_{1+\frac{\eps}{16}}(n)} 
     \geq (1+\frac{\eps^2}{128 \cdot c \log n})^{-\frac{\eps}{4} \cdot{}\frac{128 \cdot c \cdot \log(n)} {\eps^2} } \\
     (1+\eta)^{-k}  \geq  EXP\left[-\frac{\eps}{4}  \right] \geq  EXP\left[-\log(1+\frac{\eps}{2})\right] \geq (1+\frac{\eps}{2})^{-1}
\end{align*}
\end{proof}

\approx*

\begin{proof}
\textbf{First, we show that for all vertices $v$, $\gee(v) \leq (1 + \eps)\rho^*(v)$.}

        Suppose for the sake of contradiction that there exists a vertex $u$ with $\gee(u) > (1+ \eps) \rho^*(u)$.
        We fix $\rho^*(u) = \gee^*(u) = \gamma$ and work with $\gamma$ throughout the remainder of this proof to show a contradiction. 
     By Corollary~\ref{cor:alwaysfair}, there exists at least one locally fair fractional orientation $\overrightarrow{G}_x$. 
     By Corollary~\ref{cor:equal}, every vertex $v$ in this orientation has out-degree $\gee_x(v) = \gee^*(v) = \rho^*(v)$. 
    And thus, the fractional orientation $\overrightarrow{G}_x$ is not equal to $\overrightarrow{G}$. 
    
    Given $\overrightarrow{G}$ and a locally fair fractional orientation $\overrightarrow{G}_x$, we do three steps:
    \begin{itemize}
        \item We partition the vertices of $G$ to create two graphs $G_1$ and $G_2$. 
            The partition is based on the orientation $\overrightarrow{G}_x$ as we set: $G^1 = G[ v \in V \mid \gee_x(v) \leq \gamma ]$ and by $G^2 = G[v \in V \mid \gee_x(v) > \gamma]$. For ease of notation, we write any edge with one endpoint $a \in V(G^1)$  and one endpoint $b\in V(G^2)$  as $(a, b)$ and never as $(b, a)$. 
            \item From $G^1$, we create a family of nested subgraphs using $\overrightarrow{G}$. We define graphs $G^1_{i} := G^1[ v \in V(G^1) \mid  \gee(v) \geq \frac{\gee(u) }{(1 + \eta)^i}]$. 
    We denote by $k$ the lowest integer such that $|V(G^1_{k+1})| < (1 + \frac{\eps}{16}) |V(G^1_{k})|$.
We apply Lemma~\ref{lem:estimate} to observe that $(1 + \eta)^{-k} \geq (1 + \eps/2 )^{-1}$.
\item Finally, we use both orientations to create three claims that contradict one another.           
    \end{itemize}

\noindent
    \emph{The first claim. }
     We denote by $E_\uparrow$ the set of all edges $e = (a, b)$ with $a \in V(G^1_{k+1})$ and $b \in V(G^2)$ and claim that:
\begin{equation}
\label{eq:claima}
\sum_{e \in E(G^1_{k+1}) }  \gee(e) + \sum_{e \in E_\uparrow} \gee(e) \leq \sum_{v \in V(G_{k+1}^1)} \gee_x(v). 
\end{equation}

\noindent
    Indeed for $\overline{ab} \in E(G^1_{k+1})$, both endpoints are in $V(G_{k+1})$. Because $\overrightarrow{G}_x$ is locally fair, this implies that $\gee_x(a \tor b) + \gee_x(b \tor a) = \gee( \overline{ab})$.
    For all $\overline{ab} \in E_\uparrow$, $\gee_x(b) > \gamma \geq \gee_x(a)$ per definition of $G^1$ and $G^2$. By local fairness, $\gee_x(a \tor b) = \gee(\overline{ab})$ and the inequality follows.

\noindent
\emph{The second claim.}
    We secondly claim that:
\begin{equation}
\label{eq:claimb}
      \sum_{v \in V(G^1_k)} \gee(v) > \sum_{v \in V(G_{k+1}^1)} \gee_x(v).
\end{equation}

\noindent
This is because we can lower bound $ \sum_{v \in V(G^1_k)} \gee(v)$ as follows:

    \begin{align*}
    \sum_{v \in V(G^1_k)} \gee(v) &\geq (1 + \eta)^{-k} \cdot \gee(u) \cdot |V(G^1_k)| > (1 + \eta)^{-k} \cdot \gee(u) \cdot |V(G^1_{k+1})| \cdot (1 + \frac{\eps}{16})^{-1} \\
    &> (1 + \frac{\eps}{2})^{-1} \cdot (1 + \eps) \gamma \cdot |V(G^1_{k+1})| \cdot (1 + \frac{\eps}{16})^{-1} \geq \gamma \cdot |V(G^1_{k+1})|
\end{align*}
The claim follows by noting that per definition of $G^1$, for all $v \in V(G^1_{k+1})$, $\gee_x(v) \leq \gamma$. 

\noindent
    \emph{The third claim.} Lastly, we claim that:
    
\begin{equation}
\label{eq:claimc}
    \sum_{v \in V(G^1_k)} \gee(v) \leq \sum_{e \in E(G^1_{k+1}) }  \gee(e) + \sum_{e \in E_\uparrow} \gee(e)
    \end{equation}
    Consider any $v \in V(G^1_k)$ and any vertex $a$ with $\gee(v \tor a) > 0$.
    Recall that $\overrightarrow{G}$ is an $\eta$-fair orientation.
    Thus, if $a \in G^1$, then $\overline{va} \in E(G^1_{k+1})$. 
    If $a \in G^2$, then per definition $\overline{va} \in E_\uparrow$. Per definition of a fractional orientation $\gee(v \tor a) \leq \gee( \overline{va} )$ and so the claim follows. 

\noindent
\emph{A contradiction.}
    Equation~\ref{eq:claima}, \ref{eq:claimb} and \ref{eq:claimc} contradict each other.  Thus, we have proven that for all vertices $v$, $\gee(v) \leq (1 + \eps) \rho^*(v)$.

\textbf{Next, we show that for all vertices $v$, $\gee(v) \geq (1 + \eps)^{-1} \rho^*(v)$}
    Assume for the sake of contradiction that there exists a vertex $u$ with $\gee(u) < (1 + \eps)^{-1} \rho^*(u)$. 
    We use the local density (Definition~\ref{def:diminishing}+\ref{def:localdensity}). 
    I.e., there exists a unique integer $i$ and two sets of vertices $(S_i, B_{i-1})$ with $u \in S_i$.
    If we denote $B = B_{i-1}$ then:
    \[
    S_i = arg  \, \max\limits_{X \subset V - B} \frac{1}{|X|} \sum_{ \hat{E}_B(X) } \gee(e ) \textnormal{ and } \rho^*(u) = \frac{1}{|S_i|} \sum_{ \hat{E}_B(S_i) } \gee(e).
    \]

    \noindent
    Next, we use our $\eta$-fair fractional orientation to define a family of vertex sets:
    \[
    H_j := \{ v \in S_i \mid \gee(v) \leq \gee(u) \cdot (1 + \eta)^j \}.
    \]

    \noindent
    For the minimum $k$ where $|H_{k+1}| <  (1 + \frac{1}{16} \eps) |H_{k}|$. 
    $(1 + \eta)^{k} \leq (1 +  \frac{\eps}{2})$ (Lemma~\ref{lem:estimate}).

    Denote by $E_\uparrow$ all edges with one vertex in $H_{k}$ and the other in $S_i -H_{k}$. Recall that $\hat{E}_B(X)$ denotes all edges that have one endpoint in $X$ and the other in $X$ or $B$. 
    Since $H_k \cap B = \emptyset$, $\hat{E}_B(S_i - H_k) = \hat{E}_B(S_i) - (\hat{E}_B(H_k) \cup E_\uparrow)$ and so:
    \begin{equation}
        \label{eq:localdensity_equality}
        \hat{\rho}_B(S_i - H_k) \cdot (|S_i| - |H_k|) = \sum_{e \in \hat{E}_B(S_i - H_k) } \gee(e) =  \left( \sum_{e \in \hat{E}_B(S_i) } \gee(e) \right) - \left(   \sum_{e \in \hat{E}_B(H_k) \cup E_\uparrow} \gee(e) \right)
    \end{equation}

    We claim that:
    \begin{equation}
        \label{eq:v-estimate}
            \sum_{e \in \hat{E}_B(H_k) \cup E_\uparrow} \gee(e) \leq \sum_{v\in V(H_{k+1})} \gee(v). 
    \end{equation}

    Indeed:
    \begin{itemize}
        \item if $\overline{ab} \in E(H_k)$ then $\gee(\overline{ab})$ is distributed over $\gee(a)$ and $\gee(b)$ with $a, b \in H_k \subset H_{k+1}$.
        \item if $\overline{ab} \in E_\uparrow$ with $a \in H_k$ and $b \in S_i \cap H_{k+1}$, then then $\gee(\overline{ab})$ is distributed over $\gee(a)$ and $\gee(b)$ with $a, b \in H_{k+1}$.
        \item if $\overline{ab} \in E_\uparrow$ with $a \in H_k$ and $b \in S_i - H_{k+1}$, then $\gee(b) > (1+\eta)\gee(a)$ and by definition of $\eta$-fairness, $\gee(a \tor b) = \gee(\overline{ab})$.
    \end{itemize}

    Now:
    \begin{align*}
        &\sum_{v\in V(H_{k+1})} \gee(v) \leq \gee(u) \cdot (1 + \eta)^{k+1} \cdot |V(H_{k+1})| \leq \gee(u) \cdot (1 + \frac{\eps}{2}) \cdot |H_{k}| \cdot (1 + \frac{\eps}{16})  \Rightarrow \\
        &\sum_{v\in V(H_{k+1})} \gee(v) \leq (1+\eps)\gee(u) \cdot |H_{k}| < \rho^*(u) \cdot |H_{k}| \Rightarrow  \\
       &    \sum_{e \in \hat{E}_B(H_k) \cup E_\uparrow} \gee(e)  < \rho^*(u) |H_k|  \quad  \quad \quad  \quad  \quad \quad  \emph{by Equation~\ref{eq:v-estimate}}
\end{align*}

\noindent
We now apply the fact $\rho^*(u)|S_i| =  \sum\limits_{ e \in \hat{E}_B(S_i) } \gee(e)$ to note that:
\[
\sum_{e \in \hat{E}_B(S_i - H_k) } \gee(e) =   \sum_{e \in \hat{E}_B(S_i) } \gee(e) \quad   - \quad  \sum_{e \in \hat{E}_B(H_k) \cup E_\uparrow} \gee(e) > \rho^*(u)|S_i| - \rho^*(u) |H_k|
\]

However, this implies that the set $S_i - H_k$ is not empty, since $S_i = H_k$ would imply that $0 > 0$.
Therefore, we may apply Equation~\ref{eq:localdensity_equality} to claim that: 

\[
\hat{\rho}_B(S_i - H_k) \cdot (|S_i| - |H_k|) > \rho^*(u) (|S_i| - |H_k|)  \Rightarrow \hat{\rho}_B(S_i - H_k) > \rho^*(u).
\]

\noindent
However, we now have found a non-empty set $X = S_i - H_k \subseteq V - B$, where $\hat{\rho}_B(X) > \hat{\rho}_B(S_i)$ which contradicts the definition of $S_i$. 
\end{proof}

If $G$ is a unit-weight graph, Chekuri et al.~\cite{chekuri2024adaptive} present a dynamic algorithm to maintain an $\eta$-fair orientation in a unit-weight graph with $\eta \in O(\eps^{-2} \log n)$. Thus, by Theorem~\ref{thm:approx}, we may now conclude that they approximate the local density (and/or the local out-degree):

\dynamic*

\section{Results in LOCAL}
\label{sec:local}

\noindent
We prove that the local out-degree of each $v \in V$ is (largely) determined by its local neighbourhood. As a result, we immediately get an algorithm to solve Problem~\ref{newproblem:value} in LOCAL. 

\newlocal*

\begin{proof}
    To prove the theorem, we design a simple deletion-only algorithm to maintain an $\eta$-fair orientation. For $\eta = \frac{\eps^{2}}{128 \log n}$, this algorithm has a recursive depth of $O(\eps^{-2} \log^2 n)$. 
    
    Specifically, we say that a directed edge $\overline{uv}$ is bad whenever $\gee(u \tor v) > 0$ and $\gee(u) > (1 + \eta)\gee(v)$. 
    In an $\eta$-fair orientation no edge is bad.     
    Whenever we delete an edge $\overline{x_1 x_0}$, the out-degree $\gee(x_1)$ decreases by $\gee(x_1 \tor x_0)$. For vertices $x_2$ with $\gee(x_2 \tor x_1) > 0$, it may now be that $\gee(x_2) > (1 + \eta)\gee(x_1)$ (and thus the edge $\overline{x_2 x_1}$ becomes bad). Note if there exists such a vertex $x_2$, then it must hold for the vertex $x_2^* := \arg \max_{x_2 \in V} \{ \gee(x_2) \mid \gee(x_2 \tor x_1) > 0 \}$. 
    This leads to a recursive algorithm to decrease the out-degree of a vertex (Algorithm~\ref{alg:recursive}). 
    
 \begin{algorithm}[t]
    \caption{\textsc{decrease}($\gee(u \tor v)$ by $\Delta$ -- assuming that $\Delta \leq \gee(u \tor v)$)}
    \label{alg:recursive}
    \begin{algorithmic}
    \STATE $w^* \gets \arg \max_{w \in V} \{  \gee(w) \mid \gee(w \tor u) > 0 \}$
    \WHILE{ $\Delta > 0$ \textbf{and}  $\gee(w^*) > (1 + \eta)( \gee(u) - \Delta)$}
    \IF{ $\Delta > \gee(w^* \tor u)$}
    \STATE \textsc{Decrease}($\gee(w^* \tor v)$ by $\gee(w^* \tor v)$)
    \STATE $\Delta = \Delta - \gee(w^* \tor u)$
    \ELSE
    \STATE \textsc{Decrease}($\gee(w^* \tor v)$ by $\Delta$)
    \ENDIF
    \STATE $w^* \gets \arg \max_{w \in V} \{  \gee(w) \mid \gee(w \tor u) > 0 \}$
    \ENDWHILE
    \end{algorithmic}
  \end{algorithm}
    We claim that this algorithm as a recursive depth of $O(\log_{1 + \eta} n)$.
    Indeed any sequence of recursive calls is a path in $G$. Denote the path belonging to the longest sequence of recursive calls by $x_0, x_1, \ldots x_\ell$. 
    Since $\Delta$ is always at most 1, it must hold for all $i$ that: $\gee(x_i) > (1 + \eta)(\gee(x_{i-1}) - 1))$. 
    Since a graph may have at most $n^2$ edges, $\gee(x_\ell) \leq n$ and it follows that the recursive depth is at most $\ell \in O(\log_{1+\eta} n)$.
 We now apply $\log(1 + x) \geq x / 2$ and note that:
$\ell \leq \frac{\log n}{\log(1 + \eta)} \leq \frac{\log n}{\eta / 2} \subseteq O( \frac{\log n}{ 64 \eps^2 / \log n}) \subseteq O(\frac{\log^2 n}{\eps^2})$.   

Given this theoretical algorithm, we prove the lemma. Consider any fair orientation $\overrightarrow{G}$. Then by Theorem~\ref{thm:localDensity} for any vertex $v$, $\gee(v) = \rho^*(v)$. 
Choose some $k \in \Theta(\eps^{-2} \log^2 n)$ sufficiently large and let $H_k(v)$ be the $k$-hop neighbourhood of $v$ and $E_k$ be all the edges in $H_{k+1}(v)$ that are not in $H_{k}$. 

Choose $\eta = \frac{\eps^{2}}{128 \log n}$. The orientation $\overrightarrow{G}$ is per definition an $\eta$-fair orientation. We run on $\overrightarrow{G}$ our deletion-only algorithm, deleting all edges in $E_k$. 
Since our algorithm has a recursive depth of $\ell < k$, we end up with an $\eta$-fair orientation of $H_k(v)$ where $\gee(v) = \rho^*(v)$. We apply Theorem~\ref{thm:approx} to conclude that $\rho^*(v) = \gee(v) \in [(1+\eps)^{-1} \rho_k^*(v), (1+\eps)^{-1} \rho_k^*(v)]$ which concludes the theorem.
\end{proof}

\local*

\section{Results in CONGEST}
\label{sec:comp}

We now describe an algorithm in CONGEST that for any unit-weight graph $G$, creates an $\eta$-fair orientation (with $\eta = \frac{\eps^2}{128 \cdot \log n}$). 
Our algorithm uses as a subroutine a distributed algorithm to compute a \emph{blocking flow} in an $h$-layered DAG in $O(\textnormal{Blocking}(h, n))$ rounds. 

\begin{definition}
    An edge-capacitated DAG $G$ is $h$-layered if the vertices can be embedded on a grid of height $h$, such that every directed edge $\overline{uv}$ points downwards. 
\end{definition}

\begin{definition}
    For an edge-capacitated DAG $G$ with sources $S$ and terminals $T$, a flow $f$ from $S$ to $T$ is \emph{blocking} if every augmenting path of $f$ contains at least one saturated edge. 
\end{definition}

\begin{lemma}[Lemma 7.2 and 9.1 in~\cite{haeupler2023maximum}]
    \label{lem:blocking}
    There exists an algorithm which, given an $n$-node $h$-layer edge-capacitated DAG $D$ with sources $S$ and terminals $T$ computes a blocking ST-flow in CONGEST in:
    \begin{itemize}
        \item $\textnormal{Blocking}(h, n) = \tilde{O}(h^4)$ rounds with high probability,
        \item $\textnormal{Blocking}(h, n) = \tilde{O}(h^6 \cdot 2^{c\sqrt{\log n}})$ deterministic rounds for some constant $c$. 
    \end{itemize}
\end{lemma}

We compute an $\eta$-fair orientation by repeatedly constructing a DAG with $h \in O(\eps^{-2} \log^2 n)$ and computing blocking flows. 
Theorem~\ref{thm:approx} implies in an $\eta$-fair orientation (for our choice of $\eta$) the out-degree of each vertex $v$ is a $(1+\eps)$-approximation $\rho^*(v)$.

\subparagraph{The initialising step.}
Before we start our algorithm, we create a starting orientation where we set for every edge $e = \overline{uv}$, $\gee(u \tor v) = \frac{1}{2} \gee(e)$ and $\gee(v \tor u) = \frac{1}{2} \gee(e)$. This gives each vertex $u$ \emph{some} {out-degree} $\gee(u)$ which we partition:

\begin{definition}
    Let each vertex $u$ have out-degree $\gee(u)$.
    We define \emph{level} $i$ as:
    \[
    L_i := \left\{ u \in V \mid \gee(u) \in \left[ (1 + \frac{\eta}{2})^i, (1 + \frac{\eta}{2})^{i+1} \right] \right\}. 
    \]
    A vertex $u \in L_i$ is \emph{at level} $i$ and $\ell'$ denotes the highest level that is not empty.  \newline
    Whenever $\gee(u) = (1 + \frac{\eta}{2})^i$, $u$ may decide whether $u \in L_i$ or $u \in L_{i-1}$; whenever our algorithm increases $\gee(u)$ the vertex $u$ defaults to the lowest possible level and vice versa.  
\end{definition}

\begin{definition}
    Consider an edge $\overline{uv}$ with $u \in L_i$ and $v \in L_j$. We say that:
    \begin{itemize}
        \item $(u, v)$ is an \emph{out-edge} from $u$ and an in-edge into $v$ whenever $i > j$ and  $\gee(u \tor v) > 0$, and
        \item $(u, v)$ is \emph{violating} whenever $i > j + 1$ and $\gee(u \tor v) > 0$
    \end{itemize}
\end{definition} 
\noindent
Note that the orientation is $\eta$-fair whenever there exist no violating edges.

\begin{figure}[t]
  \centering
  \includegraphics[width = \linewidth]{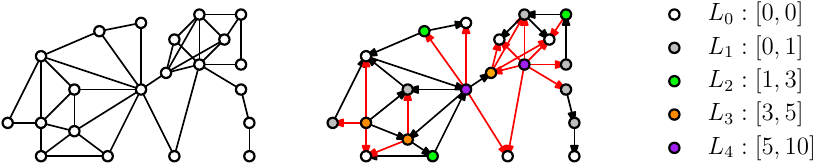}
  \caption{
Given a graph $G$, we arbitrarily orient $G$. 
This allows us to partition the vertices of $G$ into \emph{levels} $L_1, \ldots L_6$ based on their current out-degree. 
We say that an edge $(u, v)$ is \emph{violating} whenever $\gee(u \tor v) > 0$, $u \in L_i$, $v \in L_j$ and $i > j+1$. We show violating edges in red. Our algorithm iterates over an integer $h$ from high to low, and tries to flip all violating edges from level $L_h$. 
  }
  \label{fig:levelpartition}
\end{figure}

\begin{lemma}
\label{lem:elbound}
Let $\eta \leq \frac{\eps^2}{128 \cdot \log n}$.
For our orientation, let $\ell'$ be the highest level such that $L_{\ell'}$ is not empty then $\ell'\leq \eps^{-2} \log^2 n $. 
\end{lemma}

\begin{proof}
The maximal out-degree of a vertex is $n$.  
Thus, $\ell' \leq  \frac{\log n}{\log(1 + \frac{\eta}{2})}$.
We now apply $\log(1 + x) \geq x / 2$ and note that:
$\ell' \leq \frac{\log n}{\log(1 + \frac{\eta}{2})} \leq \frac{\log n}{\eta / 4} = \frac{\log n}{ 32 \eps^2 / \log n} \leq \frac{\log^2 n}{\eps^2}$.
\end{proof}

\subsection{Algorithm overview}

Denote $\ell = \lceil \eps^{-2} \log^2 n \rceil$.
Our algorithm runs on a `clock' denoted by $(h : m : s)$ where:
\begin{itemize}
       \item Each \textsc{hour} $h$ lasts $2 \lceil \eta^{-1} \rceil + 2$  \textsc{minutes},
           \item Each even \textsc{minute} $m$ lasts $4 \lceil \log_{8/7} n \rceil + 2$ rounds,
           \item Each odd \textsc{minute} $m$ in \textsc{hour} $h$ lasts $\ell - h + 1$ \textsc{seconds},
               \item Each \textsc{second} $s$ lasts $\ell + \textnormal{Blocking}(\ell, n)$ rounds.
\end{itemize}

\noindent 
Each vertex tracks the clock to know which actions of our algorithm it should execute (if any). 
Our clock is special, in the sense that hours tick downwards.  Minutes and seconds tick upwards, starting from zero.  
Each vertex $v$ in our graph keeps track of the current time, measured in the current \textsc{hour} $h$, \textsc{minute} $m$ and \textsc{second} $s$. 
We maintain the following:

\begin{invariant}
\label{inv:hour}
  During \textsc{hour} $h$, there are no violating out-edges from level $k$ for all $k > h + 1$. 
   At the start of each even \textsc{minute} in \textsc{hour} $h$, there are no violating out-edges from level $k$ for all $k > h$.
\end{invariant}

\noindent
This invariant implies that we compute an $\eta$-fair orientation when the clock reaches $(0 : 0 : 0)$. Going from \textsc{hour} $h$ to \textsc{hour} $h-1$ we maintain this invariant by flipping directed paths:

\begin{definition}
    For any edge $\overline{uv}$ with $\gee(u \tor v) > 0$, we say that our algorithm is \emph{flipping} $\overline{uv}$ whenever it decreases $\gee(u \tor v)$ (increasing $\gee(v \tor u)$). Moreover, we say that $\overline{uv}$ is flipped whenever our algorithm has decreased $\gee(u \tor v)$ such that $\gee(u \tor v) = 0$. 
\end{definition}

\subparagraph{Algorithm (see also Figure~\ref{fig:algorithm_overview} and Algorithms~\ref{alg:hour} and \ref{alg:evenminute} and \ref{alg:minute} and \ref{alg:second})}
Each time frame has a purpose:
\begin{itemize}
    \item Each \textsc{hour} $h$, the goal is to identify and `fix' all violating out-edges from vertices in $L_h$; without introducing violating out-edges from vertices in a level $L_k$ with $k > h$. We do this iteratively, using two different steps:
    \begin{itemize}
        \item Each even \textsc{minute}, we fix all violating out-edges from vertices in $L_h$, possibly creating violating out-edges from vertices in $L_{h+1}$ to vetices in $L_{h-1}$. 
        \item Each odd \textsc{minute}, we fix all violating out-edges from vertices in $L_{h+1}$ to vertices in $L_{h-1}$. We do this in such a manner, that we create no new violating out-edges from vertices in $L_k$ for $k > h$. However, we may create violating out-edges from level $L_h$.
    \end{itemize}
    \item Each even \textsc{minute} $m'$, we define the set $V_{m'} \subset L_h$ of vertices that have at least one violating out-edge.  
    Over $4 \lceil \log_{8/7} n \rceil + 2$ rounds, each $u \in V_{m'}$ flips violating out-edges $\overline{uv}$. Moreover:
        \begin{itemize}
            \item The out-degree $\gee(u)$ is decreased such that $u$ drops at most one level, and 
           \item The out-degree $\gee(v)$ such that $v$ increases its level up to at most $h - 1$. 
           \end{itemize}
        We prove that at the end of this \textsc{minute} $m'$, there exist no more $u \in V_{m'}$ with a violating out-edge. 
        Thus, for each vertex $u \in V_{m'}$, either $u$ decreased one level, or all $\overline{uv}$ that were violating now have that $\gee(u \tor v)  = 0$, or $v \in L_{h-1}$. 
        \item In each odd \textsc{minute} $m$, we inspect the vertices $T_m := \{ u \in V_{m-1} \textnormal{ that dropped a level} \}$. 
                For each $u \in T_m$, for each edge $\overline{wu}$ with $w \in L_{h+1}$ and $\gee(w \tor u) > 0$ the edge $(w, u)$ has become a violating in the previous \textsc{minute}. 
                We want to fix these edges, whilst guaranteeing that we create no violating in-edges from vertices in level $h+2$ and above.   \\
                We obtain this, by recording at the start of the \textsc{minute} for every vertex $u$ its level $l_m(u)$. 
                We then invoke $\ell - h + 1$ \textsc{seconds}.    Each \textsc{second} $s$, we create a DAG $D_s$ where the sinks are $T_m$. We increase for all $u \in T_m$ the out-degree $\gee(u)$ by flipping a directed path from a source in $D_s$ to $T_m$. We construct our DAG in such a manner that this procedure does not create new violating in-edges, and that for all \textsc{seconds} $s$, $D_{s+1}$ is a subgraph of $D_s$. 
    \item At the start of \textsc{second} $s$ of each \textsc{minute} $m$, we fix for every edge $\overline{uv}$ the values $\gee_s(u \tor v)$, and the levels $l_s(u)$ and $l_s(v)$ of $u$ and $v$. We then define a graph $D_s = (V_s, E_s)$ as follows: 
        \begin{itemize}
            \item The edges $E_s$ are all violating in-edges to vertices in $T_m$ plus all $\overline{uv}$ with:  \\
$
            l_s(u) = l_m(u) > h + 1$, $l_s(u) > l_s(v)$, and $\gee_s(u \tor v) > 0$. 
            \item The vertices $V_s$ include $T_m$ plus all vertices in $G$ with a directed path to $T_m$ in $E_s$. 
            \item The vertices $S_s \subset V_s$ are all sources in $D_s$ (these are not in $T_m$).  
        \end{itemize}
        For each $v \in T_m$ we increase $\gee(v)$ by flipping a directed path from a vertex in $S_s$ to $v$. 
        We continue this until either $\gee(v) = (1 + \frac{\eta}{2})^{h+1}$, or, there exist no more edges $(u, v) \in E_s$ with $\gee(u \tor v) > 0$. In both cases, $v$ has no more violating in-edges. 
        To find these directed paths, we create a flow problem on a graph $D_s^*$ where the maximal path has length $\ell + 2$:
        \begin{itemize}
            \item  For each $u \in S_s$, we define $\sigma(u) = \gee(u) - (1 + \frac{\eta}{2})^{l_m(u) - 1}$ (the maximal amount we can decrease $\gee(u)$ such that is does not arrive in level $l_m(u) - 2$). We connect every $u \in S_s$ to a unique source $s_u$ where the edge $\overline{s_u u}$  has capacity $\sigma(u)$.
            \item For each $v \in T_m$, we define $\delta(v) = (1 + \frac{\eta}{2})^{h+ 1}  - \gee(v)$ (the maximal amount we can increase $\gee(v)$ such that it does not arrive in level $h + 1$). We connect every $v \in T_m$ to a unique sink $t_v$ where the edge $\overline{v t_v}$ has capacity $\delta(v)$. 
            \item Each other edge $\overline{uv} \in E_s$ has capacity $\gee(u \tor v)$. 
        \end{itemize}
\end{itemize}

\begin{figure}[h]
  \centering
  \includegraphics[width = \linewidth]{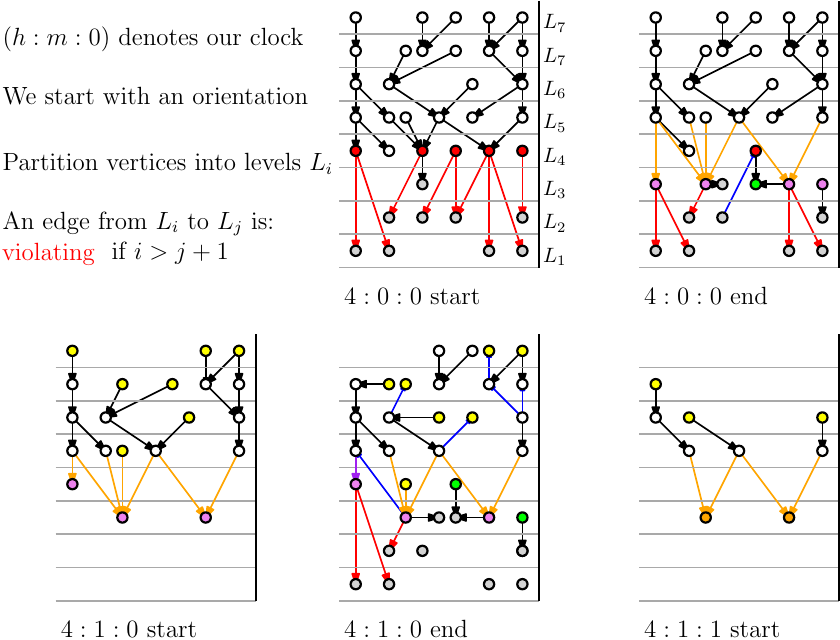}
  \caption{
$(4 : 0 : 0)$ - at the first \textsc{minute} start in \textsc{hour} $4$, we consider all violating out-edges $\overline{uv}$ from level $4$ (red). Per definition, these edges point to level $2$ or lower.  \newline
$(4 : 0 : 0)$ - at the first \textsc{minute} end, either $u$ has dropped a level (pink), $v$ increased their level to $h - 1$ (green) or the edge $\overline{uv}$ is flipped (blue). We consider violating in-edges to pink vertices (orange)  \newline
$(4 : 1 : 0)$ - at the first \textsc{second} start, we construct a DAG $D_0$ where the edges are the orange edges plus black edges. The vertex set are all $u$ with a directed path to a pink vertex ($ v \in T_1$). The yellow vertices are $S_0$. \newline
$(4 : 1 : 0)$ -  at the first \textsc{second} end, edges in the DAG may have flipped (blue), vertices in $S_0$ may have dropped a level or vertices in $T_1$ may have increased a level (making some edges no longer violating -- purple).  \newline
$(4 : 1 : 1)$ - at the second \textsc{second} start, we construct a DAG $D_1$. Note that $D_1$ is a subgraph of $D_0$. \vspace{-0.5cm}
  }
  \label{fig:algorithm_overview}
\end{figure}

\subsection{Formal algorithm definition.}

We now formalise our algorithm top-down, starting with defining variables.

\begin{definition}
    At the start of $(h, m', 0)$, where $m'$  is even, we fix the set
    $
    V_{m'} := \{  v \in L_h \mid v \textnormal{ has at least one violating out-edge} \}. $
    At the start of $(h, m, 0)$  where $m$  is odd, we fix:
    \begin{align*}
    T_m := \{  v \in V_{m-1} \mid \quad   &v \textnormal{ decreased by one level in the previous } \textsc{minute} \textnormal{ and }  \\ & v \textnormal{ has at least one violating in-edge} \}.
    \end{align*}
    Finally, we denote for any vertex $u$ by $l_m(u)$ its level \emph{at the start of the }\textsc{minute} $m$.
\end{definition}

\begin{definition}
    \label{def:variables}
    At the start of $(h : m : s)$, where $m$ is odd,  we fix the following:
\begin{itemize}
    \item For any vertex $u$, we denote by $l_s(u)$ its level {at the start of the \textsc{second}}. 
    \item For any edge $\overline{uv}$ denote by $\gee_s(u \tor v)$ the out-degree from $u$ to $v$ at the start of the \textsc{second}. 
    \item We define the edge set  $E_s$ as all violating in-edges to vertices in $T_m$ plus all $\overline{uv}$ with:  \\
$
            l_s(u) = l_m(u) > h + 1$, $l_s(u) > l_s(v)$, and $\gee_s(u \tor v) > 0$. 
    \item $V_s$ are all vertices with a directed path to a vertex in $T_m$.
    \item $D_s = (V_s, E_s)$ is a DAG where $S_s$ are all the sources (per definition $S_s \cap T_m = \emptyset$).
    \end{itemize} 
\end{definition}
\newpage

\begin{definition}
    \label{def:DAG}
    At the start of $(h, m, s)$ where $m$ is odd, given $T_m$, $S_s$ and $D_s$, we define the DAG $D_s^*$ by connecting all $u \in S_s$ to a unique sink $s_u$ and all $v \in T_m$ to a unique sink $t_v$.
    \begin{itemize}
                 \item  For $u \in S_s$, the edge $\overline{s_u u}$ has capacity $\sigma(u) = \gee_s(u) - (1 + \frac{\eta}{2})^{l_m(u) - 1}$.
            \item For $v \in T_m$, the edge $\overline{v t_v}$ has capacity $\delta(v) = (1 + \frac{\eta}{2})^{h+ 1}  - \gee_s(v)$.
            \item Each other edge $\overline{uv} \in E_s$ has capacity $\gee_s(u \tor v)$. 
    \end{itemize}
\end{definition}

\begin{observation}
      At the start of $(h : m : s)$ with $m$ odd, if every vertex $u$ is given $l_m(u)$, whether $u \in T_m$, and $h$, then we may compute all elements of Definitions~\ref{def:variables} and \ref{def:DAG} in $\ell$ rounds. 
\end{observation}

 \begin{algorithm}[H]
    \caption{\textsc{hour}($h$)}
    \label{alg:hour}
    \begin{algorithmic}
    \FOR{$m = 0$ to $2 \lceil \eta^{-1} \rceil + 2$} 
    \IF{ $m$ is even}
    \STATE \textsc{evenminute}($h$, $m$)
    \ELSE
    \STATE \textsc{oddminute}($h$, $m$)
    \ENDIF
    \ENDFOR
    \IF{$h > 0$}
    \STATE \textsc{hour}$(h-1)$
    \ENDIF
    \end{algorithmic}
  \end{algorithm}

 \begin{algorithm}[H]
    \caption{\textsc{evenminute}(int $h$, int $m$)}
    \label{alg:evenminute}
    \begin{algorithmic}
    \STATE $V_m := \{ v \in L_h \mid v \textnormal{ has at least one violating out-edge } \}$
       \FOR{ $t = 0$ to $2 \lceil \log_{8/7} n  \rceil + 1$}
    \STATE Let $A_t \subset V_m$ be the set of vertices at level $h$ with at least one violating out-edge
    \STATE Each $a \in A_t$ determines the set $E_t(a)$ of violating out-edges.
    \STATE Each $a \in A_t$ computes $\delta_t(a) = \gee(a) - (1 + \eta)^{h - 1}$ and reports $\delta_t(a) / |E_t(a)|$ across $E_t(a)$. 
    
    \tcc{next round:}
    \STATE Let $B_t$ be the set of vertices that receive at least one violating in-edge from $A_t$. 
    \STATE Each  $b \in B_t$ determines the set $I_t(b)$ of vertices that reported a value to $v$.
    \STATE Each  $b \in B_t$ sorts the $ a \in I_t(b)$ by $\delta(a)$.  
    \STATE Each  $b \in B_t$ greedily decreases $\gee(a \tor b)$ by at most $\delta_t(a)$ for $a \in I_t(b)$; until $\gee(b) = (1 + \frac{\eta}{2})^h$.      
    \ENDFOR
    \end{algorithmic}
  \end{algorithm}

 \begin{algorithm}[H]
    \caption{\textsc{oddminute}(int $h$, int $m$)}
    \label{alg:minute}
    \begin{algorithmic}
    \STATE $T_m := \{v \in V_{m-1} \mid v \textnormal{ has at least one violating in-edge } \}$
    \FOR{ $s = \ell$ down to $h$}
    \STATE \textsc{second}$(h, m ,s)$
    \ENDFOR
    \end{algorithmic}
  \end{algorithm}

 \begin{algorithm}[H]
    \caption{\textsc{second}(int $h$, int $m$, int $s$)}
    \label{alg:second}
    \begin{algorithmic}
    \STATE Compute the graph $D_s$ in $\ell$ rounds. 
    \STATE Compute the graph $D_s^*$ by adding a shared source to $S_s$ and a shared sink to $T_m$. 
    \STATE $f =$ \textsc{ComputeBlockingFlow}($D_s^*$)
    \STATE Flip all edges in $D_s$ with the corresponding flow in $f$. 
    \end{algorithmic}
  \end{algorithm}

\subsection{Sketching our algorithm's correctness.}

Per definition, our algorithm runs in $O( \eta^{-1}  \log n \cdot  \ell \cdot (\ell + \textnormal{Blocking}(\ell, n)) = O(  \eps^{-3} \log^4 n \cdot (\eps^{-2} \log^2 n + \textnormal{Blocking}(\eps^{-2} \log^2 n, n))  )$ rounds. 
What remains is to show that we maintain Invariant~\ref{inv:hour}, which implies that we compute an $\eta$-fair orientation. 

\noindent
We note that by the choice of our algorithm's variables, we have the following property:

\begin{observation}
    \label{obs:shift}
    For all times $(h : m : s)$ with $m$ odd:
    \begin{itemize}
        \item Vertices at level $k > h$ only decrease their level and by at most one (because afterwards, $l_s(u) < l_m(u)$ and thus $u$ has no out-edges in $E_s$),
    \item  vertices at level $h - 1$ only increase their level (by at most $1$), and
    \item vertices at level $k' < h - 1$ and level $h$ do not change their level. 
    \end{itemize}
\end{observation}

We use this observation to show that we maintain Invariant~\ref{inv:hour} by induction. Trivially, Invariant~\ref{inv:hour} holds at the start of $(\ell : 0 : 0)$. 
We now assume that the invariant holds at $(h : 0 : 0)$, i.e. that there are no violating out-edges from level $L_k$ with $k > h$. 
We prove that our algorithm ensures that at the start of $(h-1 : 0 : 0)$ there are no violating out-edges from level $L_{k}$ with $k \geq h$.

Moreover, we maintain the invariant that throughout \textsc{hour} $(h-1)$ there are no violating out-edges from level $L_{k'}$ with $k' > h + 1$.
We prove this in the following way:

\begin{itemize}
\item     During $(h : m : s)$ for $m$ even, our algorithm eliminates all violating edges going from $L_h$. We show that in each iteration, the number of violating edges from $L_h$ drops by a factor $7/8$  (Lemma~\ref{lem:oneeight}), since the graph has at most $n^2$ edges. This implies that after the \textsc{minute}, there are no more violating edges going from $L_h$ (Corollary~\ref{cor:odd}). However, there may now be violating out-edges from vertices in $L_{h+1}$ to vertices in $T_{m+1} \subseteq  L_{h-1}$. 
    \item During $(h : m : s)$ for $m$ odd, our algorithm eliminates all violating edges going from $L_{h+1}$ to $T_m$. We show that for all $s$, the DAG $D_{s}$ is a subgraph of $D_{s-1}$ where the height is one fewer (Lemma~\ref{lem:decreasing_graph}).   Since the height of $D_s$ is at most $\ell - h + 1$, this implies that after the \textsc{minute} $m$, the graph $D_{s}$ is empty and there are no more violating in-edges to $T_m$ (Corollary~\ref{cor:even}).

    \item Each \textsc{minute}, our algorithm alternates between having violating edges from $L_h$ and $L_{h+1}$. We show that our algorithm can alternate at most $\eta$ times before there are no violating edges from both $L_h$ and $L_{h+1}$, which implies  Invariant~\ref{inv:hour} at the start of $(h-1:0:0)$.  (Theorem~\ref{thm:main_congest}).
\end{itemize}
Invariant~\ref{inv:hour} implies that we compute an $\eta$-fair orientation at time $(0:0:0)$, thus:
\main*

\noindent
We plug in the runtime of Blocking$(h, n)$ of Lemma~\ref{lem:blocking} by  Haeupler, Hershkowitz, and Saranurak~\cite{haeupler2023maximum} to obtain the following runtime:

\begin{restatable}{corollary}{cormain}
    \label{cor:main}
There is an algorithm in CONGEST that given a unit-weight graph $G$ and an $\eps > 0$ computes an orientation $\overrightarrow{G}$ such that $\forall v \in V$, the out-degree $\gee(v) \in [(1 + \eps)^{-1} \rho^*(v), (1 + \eps)\rho^*(v) ]$ in:
\begin{itemize}
         \item $\tilde{O} (\eps^{-11} \log^{12} n )$ rounds with high probability, or 
        \item $\tilde{O}( \eps^{-15} \log^{16} n \cdot 2^{O(\sqrt{\log n})})$ deterministic  rounds.
\end{itemize}
\end{restatable}

Finally, we apply Corollary~\ref{cor:equal} which states that $\rho^{\max}(G)= \Delta^{\min}(G) = \max_v \gee^*(v)$ to conclude:

\begin{corollary}
    There is a deterministic algorithm in CONGEST that given a unit-weight graph $G$ and an $\eps > 0$, that computes an orientation $\overrightarrow{G}$ such that:
    \[
    \max_{v \in V} \gee(v) \in [(1+\eps)^{-1}\rho^{\max}(G), (1+\eps) \rho^{\max}(G)],
    \]
    in a number of rounds that is sublinear in $n$. 
\end{corollary}

\subsection{Formally proving correctness}
We formalise Algorithms~\ref{alg:hour} and \ref{alg:evenminute} and \ref{alg:minute} and \ref{alg:second} and prove that they maintain Invariant~\ref{inv:hour}. 

First, we prove in three steps that after each odd \textsc{minute}
 (i.e., at the start of each even \textsc{minute}) there are only violating out-edges from $L_h$ and below (Corollary~\ref{cor:even}).

\begin{lemma}
    \label{lem:only_h_plus_one}
    For all times $(h : m : s)$ with $m$ odd, there are no violating edges from above level $h+1$.  
\end{lemma}

\begin{proof}
    At the start of $(h : m : 0)$, there can only be violating in-edges from vertices in $L_{h+1}$ to vertices in $T_m$. 
    Consider for the sake of contradiction the first second $s > 0$, where after $(h : m : s)$ there exists an outgoing edge $(u, v)$ with $u \in L_k$ and $k > h + 2$.
Each vertex $v$ decreases its out-degree by at most one level, and thus $l_s(u) = l_s(v) + 1$, $l_s(v) = l_m(v)$ and $v$ decreased its out-degree by one during the \textsc{second} $s$.     
    Since the edge $\overline{uv}$ is violating at the end of $(h : m : s)$, it must be that $\gee_s(u \tor v) > 0$.
    For each edge $\overline{ab}$ with $l_m(a) > l_m(b)$ the out-degree $\gee(a \tor b)$ only decreases and thus $\gee_0(u \tor v) > 0$. Moreover, by Observation~\ref{obs:shift}, $l_s(v) = l_m(v) \geq h + 1$. 
    
    At the start of $(h : m : s)$, the vertex $v$ must a source in $D_s$ (since only sources decrease their out-degree). 
    However, by the definition of $E_s$, this implies that $l_s(u) \neq l_m(u)$ (else, $E_s$ must include the edge $\overline{uv}$). Thus, $l_m(u) = l_s(u) + 1$.

    However, we now get that $l_m(u) = l_s(u) + 1 = l_s(v) + 2 = l_m(v) + 2$ and $\gee_0(u \tor v) > 0$. 
    This contradicts that $s$ is the first second where there is a violating edge from level $h+2$ or above, since at $(h : m : 0)$ the edge $u \tor v$ is violating. 
\end{proof}

\begin{lemma}
    \label{lem:decreasing_graph}
    For each time $(h : m : s)$ with $m$ odd, $D_s \subset D_{s+1}$ and the height of $D_{s+1}$ is at least one fewer than the height of $D_{s}$.  
\end{lemma}

\begin{proof}
    For every edge $\overline{uv} \in E_s$, by Observation~\ref{obs:shift}, our algorithm only decreases $\gee(u \tor v)$, and only decreases the level of $u$. 
    It immediately follows from the definition of $E_s$ that if $\overline{uv}$ is in $E_s$ but not in $E_{s+1}$, then $\overline{uv} \in E_l$ for all $l < s$ and $\overline{uv} \not \in E_k$ for all $k > s$. Thus, $D_s$ is a subgraph of $D_{s+1}$. 
   
    In $(h : s : m)$ we compute a blocking flow $f$ from $S$ to $T$ and flip edges according to the flow. 
    We claim that this blocking flow implies that for all $u \in S^s$, there exists no path from $u$ to $T_m$ in $E_{s+1}$ (and thus $u \not \in V_k$ for $k > s$). 
     Indeed, for any $u \in S^s$ and any $v \in T_m$, consider an arbitrary directed path $\pi$ from $u$ to $v$ in $D_s$.     One of three (not mutually exclusive) conditions must hold. Either: 
    \begin{itemize}
        \item There exists an $\overline{ab} \in \pi$ that is saturated in $f$. Then $\gee_{s+1}(a \tor b) = 0$ and $\pi \not \subseteq E_{s+1}$.
        \item The edge $\overline{v t_v}$ is saturated in $f$. Then $\gee_{s+1}(v) = (1 + \frac{\eta}{2})^{h+1}$, and $v$ is in level $h$. This implies that $v$ has no incoming violating edges (and thus no incoming edges) in $E_{s+1}$.
        \item The edge $\overline{s_u u}$ is saturated in $f$. Then $\gee_{s+1}(u) = (1 + \frac{\eta}{2})^{l_m(u)-1}$ and $u$ is in level $l_m(u) - 1$. This implies that $u$ has no outgoing edges in $E_{s+1}$. 
    \end{itemize}
    Since $D_{s+1}$ is a subgraph of $D_s$, and no vertices in $S_s$ are in $D_{s+1}$, the height of $D_{s+1}$ is at least one fewer than the height in $D_{s}$. 
\end{proof}

\begin{corollary}
\label{cor:even}
    At the start of $(h : m : 0)$ with $m$ even, there are no violating edges from levels $k > h$. 
\end{corollary}

\begin{proof}
    Consider the previous \textsc{minute}. 
    At time $(h : m-1: 0)$ the height of $D_0$ is at most $\ell - h + 1$.
    We may now apply Lemma~\ref{lem:decreasing_graph} to conclude that that at second $s = \ell - h + 1$, $V_s$ is empty.     
    
    By Lemma~\ref{lem:only_h_plus_one}, the only violating out-edges from levels $k > h$ come from level $h+1$.    
    By Observation~\ref{obs:shift}, for all $u \in L_{h+1}$ with a violating edge $\overline{uv}$: $l_m(u) = h+1$ and $v \in T_m$ (since all vertices at a level below $h+1$ only increase their out-degree).
    Thus any violating edge $(u, v) \in E_s$ with $u \in V_s$. Since for $s = \ell - h + 1$ the set $V_s$ is empty, the lemma follows. 
\end{proof}

Next, we show that after each even \textsc{minute} (i.e., at the start of each odd \textsc{minute}), there are only violating out-edges from $L_{h-1}$ and below, and from vertices in $L_{h+1}$ to vertices in $L_{h_1}$ (Corollary~\ref{cor:odd}). 

\begin{lemma}
    \label{lem:oneeight}
    At the start of $(h : m : t)$, with $m$ even, denote by $E_t$ the number of violating edges going from $L_h$. 
    Then $|E_t| \leq \frac{7}{8} |E_{t-1}|$. 
\end{lemma}

\begin{proof}
    Let $E_t \subseteq A_t \times B_t$.
    For convenience, we denote $\Delta_h = (1 + \frac{\eta}{2})^{h} - (1 + \frac{\eta}{2})^{h-1}.$
    We say that a vertex $a \in A_t$ is \emph{satisfied} if $\gee(a) \leq (1 + \frac{\eta}{2})^{h-1} + \frac{1}{2}\Delta_h$. 
    A vertex $b \in B_t$ is \emph{satisfied} whenever $\gee(b) = (1 + \frac{\eta}{2})^{h}$.
    We say that an edge $\overline{uv} \in E_t$ is \emph{satisfied} whenever $\gee(u \tor v) = 0$ or whenever at least $u$ or $v$ is satisfied. Observe that a satisfied edge is not violating. 

    \noindent
    Fix any edge $\overline{ab} \in E_t$. At the end of iteration $t$:
    \begin{enumerate}
        \item If $\gee(a \tor b) = 0$ then the edge is satisfied. 
        \item If $\gee(a \tor b) > 0$ and $\gee(a \tor b)$ is not decreased by $\frac{\delta_t(a)}{|E_t(a)|}$. This only occurs in our code if $b$ is satisfied and thus $\overline{ab}$ is satisfied. 
        \item If otherwise $\gee(a \tor b) > 0$ and $\gee(a \tor b)$ is decreased by $\frac{\delta_t(a)}{|E_t(a)|}$, we make a case distinction:
        \begin{itemize}
            \item If $a$ is satisfied at the end of iteration $t$, then $\overline{ab}$ is satisfied. 
            \item Otherwise,       
            let $\gee(a)$ denote the out-degree of $a$ at the beginning of iteration $t$ and $\gee'(a)$ denote the out-degree of $a$ at the end of iteration $t$. 
            Per definition of $a \in L_h$, $( 1 + \frac{\eta}{2})^{h } \leq \gee(a) \leq ( 1 + \frac{\eta}{2})^{h+1}$.
            Thus, $\delta_t(a)$ is at least $\Delta_h$ and fewer than $4 \Delta_h$. 
            If follows that if $a$ was not satisfied at the end of iteration $t$, at least $\frac{1}{8}$'th of the edges $\overline{ac} \in E_t(a)$ has not been decreased by $\frac{\delta_t(a)}{|E_t(a)|}$. To all these edges $\overline{ac}$, case $2$ applies.
        \end{itemize}
    \end{enumerate}

    Since satisfied edges are no longer violating, it follows that $|E_{t+1}| \leq \frac{7}{8}|E_t|$.
\end{proof}

\begin{corollary}
    \label{cor:odd}
       At the start of $(h : m : 0)$ with $m$ odd, there are no violating edges from levels $k > h + 1$, or from level $h$. 
\end{corollary}

\begin{proof}
    The \emph{minute} is a for-loop of $2 \lceil \log_{8/7} n \rceil  +  1$ iterations, each taking two rounds. 
      Since a graph may have at most $n^2$ edges, $|E_0| \leq n^2$. We now apply Lemma~\ref{lem:oneeight} to obtain that after the \textsc{minute} there are no vertices in $L_h$ with a violating out-edge. 
    This may create vertices in $L_{h+1}$ with a violating out-edge, but since vertices in $L_h$ decreased their level by at most one, no vertices in $L_k$ with $k > h+1$ may have a violating out-edge.   
\end{proof}

Finally, we show one more helper lemma to imply our main theorem:

\begin{lemma}
    \label{lem:alternate}
  If at the start of $(h : x : 0)$ for even $x$,  a vertex $v$ has $l_{x}(v) = h$ and $v$ has an outgoing violating edge, then for all even $m < x$, $v \in T_m$. 
  \end{lemma}

\begin{proof}
        Consider any $m < x$ that is odd. 
    If at the start of $(h : m : 0)$, a vertex $v \not \in T_m$ then this can because of two reasons. Either:
    \begin{enumerate}
        \item $v$ did not decrease their level in \textsc{minute} $m-1$, or
        \item   $v$ did decrease their level in \textsc{minute} $m-1$ but then had no incoming violating edges.
    \end{enumerate}

  We show that Case 1 implies that $v$ has no violating out-edges at $(h : x : 0)$. Case 2 implies that $l_{x}(v) \neq h$ and this together implies the lemma. 

  \subparagraph{Case 1:}
Suppose that  $l_{m - 1}(v) = h$ and that at the start of $(h : m-1: 0)$ the vertex $v$ has no violating out-edges. Then $v$ does not change its level during \textsc{minute} $m - 1$ and $l_m(v) = h$. 
      During odd \textsc{minutes}, vertices at level $h$ and levels $k < h - 1$ do not change their level. 
      Thus, at the start of $(h : m+1 : 0)$, $l_{m+1}(v) = h$ and $v$ has not gained any violating out-edges. We recursively apply this argument to conclude that $l_{x}(v) = h$ and $v$ has no violating out-edges. 

      \subparagraph{Case 2:}
    Suppose that $l_m(v) = h - 1$ and $v$ has no violating in-edge at the start of $(h: m : 0)$.
Then by our algorithm, $v$ does not change its level during $\textsc{minute}$ $m$. But this means that in \textsc{minute} $m+1$ the vertex is in level $L_{h-1}$. During even \textsc{minutes}, vertices at level $h-1$ do not change their out-degree and so  $l_{m+2}(v) = h-1$. Moreover, vertices in $L_k$ with $k \geq h$ may only decrease their out-degree and thus cannot suddenly gain a violating in-edge to $v$. Thus, $l_{m+1}(v) = h - 1$ and $v$ has no violating in-edge at the start of $(h: m + 1: 0)$. 
We may recursively apply this argument to conclude that $l_{x}(v) = h-1$. 
    \end{proof}

\begin{theorem}
    \label{thm:main_congest}
    Assume that Invariant~\ref{inv:hour} holds at the start of $(h : 0 : 0)$, then our algorithm maintains Invariant~\ref{inv:hour} throughout \textsc{hour} $h$ until and including the start of $(h-1 : 0 : 0)$. 
\end{theorem}

\begin{proof}
       Suppose for the sake of contradiction that at the start of $(h: 2 \lceil  \eta^{-1} \rceil + 2: 0)$ there exists a vertex $v \in L_h$ with at least one violating out-edge. 
    We may apply Lemma~\ref{lem:alternate} so that: 
    
    \begin{itemize}
        \item for even $m$, at the start of $(h : m : 0)$, $l_m(v) = h$ and $v$ had a violating out-edge
    \item for odd $m$, at the start of $(h : m : 0)$, $l_m(v) = h-1$ and $v$ had a violating in-edge. 
    \end{itemize}

    Our algorithms guarantee that at the start of $(h : m : 0)$ for $m$ even, $\gee(v) = (1 + \frac{\eta}{2})^{h + 1}$ and at the end of $(h : m : 0)$ for $m$ odd, $\gee(v) = \gee(v) = (1 + \frac{\eta}{2})^{h - 1}$. 

    Denote for all even $m$ by $E_m(v)$ the violating out-edges from $v$ and by $\Delta_m = \sum_{\overline{vw}} \gee(v \tor w)$.
    Since $v$ is in level $h$ it must be that $\Delta_0 \leq (1 + \frac{\eta}{2})^{h + 1}$. 
    When going from minute $m$ to $m+1$, the out-degree of $v$ decreases by only decreasing $\gee(v \tor w)$ for $\overline{vw} \in E_m(v)$.
    When going from minute $m+1$ to $m$, the out-degree of $v$ increases by only increasing $\gee(v \tor u)$ for $\overline{vw} \not \in E_m(v)$. 
    Thus $\Delta_m$ decreases by $(1 + \frac{\eta}{2})^{h + 1} - (1 + \frac{\eta}{2})^{h - 1}$ during \textsc{minute} $m$. 
    And, after at most
    \[
    \frac{\Delta_0 }{(1 + \frac{\eta}{2})^{h + 1} - (1 + \frac{\eta}{2})^{h - 1} } \leq \frac{(1 + \frac{\eta}{2})^{h + 1}}{(1 + \frac{\eta}{2})^{h + 1} - (1 + \frac{\eta}{2})^{h - 1}} \leq \eta^{-1} \quad \textnormal{ odd } \textsc{ minutes}
    \]
    the value $\Delta_m  = 0$ which contradicts the assumption that $v$ has an outgoing violating edge.     
\end{proof}

\noindent
Invariant~\ref{inv:hour} implies that we compute an $\eta$-fair orientation when the clock reaches $(0:0:0)$, thus:
\main*

\noindent
We plug in the runtime of Blocking$(h, n)$ of Lemma~\ref{lem:blocking} by  Haeupler, Hershkowitz, and Saranurak~\cite{haeupler2023maximum}:

\cormain*

 \newpage 
\bibliographystyle{plain}
\bibliography{refs}

 \newpage
\appendix

\section{Reporting a densest subgraph}
\label{app:reporting}

Problems~\ref{newproblem:value} and \ref{newproblem:value} are what we call the \emph{value variant} of the subgraph density problem, where the goal is to output a value. 
The value variant has a natural alternative, where the goal is to actually report a densest subgraph. 

\subsection{Related work: distributed densest subgraph reporting}

The reporting variant of the subgraph density problem, has two subvariants in the distributed model of computation.
Here, the second subvariant is by Harris~\cite{harris2019distributed}.

\begin{problem}
    \label{oldproblem:reporting}
    Given a graph $G$, after a computation, each vertex outputs a bit $b_v$. Denote by $H$ the subgraph induced by all vertices with a $1$-bit. We consider two variants where either:
    \begin{itemize}
        \item \textbf{Problem 4.1}: We require that $\rho(H) \leq \rho^{\max}(G) \leq (1 + \eps) \rho(H)$, or
        \item  \textbf{Problem 4.2}: Given an oracle that before outputting the bit, gives each vertex a value $\tilde{D} \leq \rho^{\max}(G)$. Each vertex outputs a bit $b_v$ and $\rho(H) \geq (1 - \eps) \tilde{D}$. 
If $\tilde{D} > \rho^{\max}(G)$ then every vertex may output $0$. 
    \end{itemize}
\end{problem}

The reporting subvariants are considerably harder than the value subvariants. 
For Problem \ref{oldproblem:reporting}.1, there is a trivial lower bound of $\Omega(d)$. 
In LOCAL, the problem is trivial to solve in $O(d)$ time. 
Su and Vu~\cite{SuVu20} can solve this variant in $O(\eps^{-1} \log n)$ rounds in LOCAL using the same observation as before and the corresponding trivial algorithm.

In CONGEST, Das Sarma et al.~\cite{SarmaLNT12} present the currently best deterministic algorithm that can report a $(2 + \eps)$-approximate densest subgraph in $O(d \cdot \eps^{-1} \log n)$ rounds (w.h.p). 
Su and Vu's~\cite{SuVu20} randomised algorithm in CONGEST can solve Problem~\ref{oldproblem:reporting}.1 with $O(d + \eps^{-4}\log^4 n)$ rounds (with high probability).
For Problem \ref{oldproblem:reporting}.2, the oracle avoids the $\Omega(d)$ lower bound and 
Su and Vu's~\cite{SuVu20} randomised algorithm can solve this problem in $O(\eps^{-4}\log^4 n)$ rounds (with high probability). See also Table~\ref{tab:results2} for an overview. 

\subsection{Reporting a locally densest subgraph}

For the local density, we can define an equivalent problem definition to Problem~\ref{oldproblem:reporting}:

\begin{problem}
    \label{newproblem:reporting}
    Given $(G, \eps)$ and a vertex $v$, each vertex outputs a bit. Denote by $H$ the 1-bit induced subgraph. We require that $v \in H$, and that $\rho(H) \in [(1-\eps) \rho^*(v), (1+\eps)\rho^*(v)]$.
\end{problem}

We design a set of algorithms to answer the reporting variant for the local subgraph density problem (see Table~\ref{tab:results2} for an overview)

\begin{table}[h]
    \centering
    \begin{adjustbox}{max width=1.3\textwidth,center}
    \begin{tabular}{c|c|c|c|c}
        Model & Problem & Each $v$ reports a bit s.t. for the $1$-bit induced graph $H$: & Rounds & Source \\
        & & (Our new result additionally requires an input vertex $v$) & & \\ \hline
L & \ref{oldproblem:reporting}.1 & $\rho(H) \in [(1+\eps)^{-1} \rho^{\max}(G), (1+\eps)\rho^{\max}(G) ]$ & $\Theta(D)$ & \cite{SuVu20} \\
 & \ref{oldproblem:reporting}.2 & $\rho(H) \geq \tilde{D} \textnormal{ where } \tilde{D} \textnormal{ is given by an oracle.}$ & $O(\eps^{-1} \log n)$ & \cite{SuVu20} \\
 & \textbf{\ref{newproblem:reporting}} & $\pmb{\rho(H) \in [(1+\eps)^{-1} \rho^*(v), (1+\eps)\rho^* (v) ]}$ & $\pmb{O(\eps^{-2} \log^2 n)}$ & \textbf{Cor.~\ref{cor:local}}      \\ \hline
C & \ref{oldproblem:reporting}.1 &  $\rho(H) \in [(2+\eps)^{-1} \rho^{\max}(G), (2+\eps)\rho^{\max}(G) ]$ & $O(D \cdot \eps^{-1} \log n)$ & \cite{SarmaLNT12} \\ 
 & \ref{oldproblem:reporting}.2 & $\rho(H) \geq \tilde{D} \textnormal{ where } \tilde{D} \textnormal{ is given by an oracle.}$ & \textcolor{orange}{ $O(\eps^{-4} \log^4 n)$}& \cite{SuVu20} \\
& \textbf{\ref{oldproblem:reporting}}\textbf{.2} & \pmb{$\rho(H) \in [(1+\eps)^{-1} \rho^*(v), (1+\eps)\rho^* (v) ]$} &$ \pmb{O(  \poly \{ \log n, \eps^{-1} \}) \cdot 2^{O(\sqrt{\log n})}}$  & \textbf{Thm.~\ref{thm:vgeneral}}   \\
& \textbf{\ref{newproblem:reporting}} & \pmb{$\rho(H) \in [(1+\eps)^{-1} \rho^*(v), (1+\eps)\rho^* (v) ]$} &$ \pmb{O(  \poly \{ \log n, \eps^{-1} \}) \cdot 2^{O(\sqrt{\log n})}}$  & \textbf{Thm.~\ref{thm:vgeneral}}   \\
  & \textbf{\ref{newproblem:reporting}} &   \pmb{$\rho(H) \in [(1+\eps)^{-1} \rho^*(v), (1+\eps)\rho^* (v) ]$} & \textcolor{orange}{$\pmb{O(\poly \{ \log n, \eps^{-1} \})}$} & \textbf{Thm.~\ref{thm:vgeneral}}    \\
    \end{tabular}
    \end{adjustbox}
    \caption{Our results in LOCAL (L) or CONGEST (C) for the reporting variant of our problem. $D$ denotes the diameter.  Orange running times are not deterministic and occur with high probability.   }
    \label{tab:results2}
\end{table}

\subsection{Our algorithm}
So far our algorithm is designed to simply output an estimate of the local density at every vertex. 
It is easily extended to output an approximate value of the maximum subgraph density in $O(\text{diameter})$ many rounds, by computing an $\eta$-fair out-orientation (for the correct value of $\eta$) and subsequently broadcasting the maximum fractional out-degree which provides the desired approximate value of $\rho^{\max}(G)$.
We can also extend the algorithm for computing an $\eta$-fair orientation to one that outputs a subgraph with no smaller than $(1-\eps)\tilde{D}$, for some input parameter $\tilde{D}\leq \rho^{\max}(G)$. 
Specifically, we show how to output $0$ or $1$ at each vertex such that the graph induced by all vertices outputting $1$ forms a subgraph with density at least $(1-\eps)\tilde{D}$.

Our reduction relies on the following generalisation of the work by Su and Vu~\cite{SuVu20}. Su and Vu showed their version of the lemma by appealing to network decompositions. We give a simple argument without appealing to network decompositions.
\begin{lemma} \label{lem:local}
    Let $1\geq \eps > 0$ and $n \in \mathbb{N}$ be given. Then for $t = \lceil \frac{2\log n}{\varepsilon} \rceil$ we have: For any graph $G$ on $n$ vertices and for all vertices $v \in V(G)$ there exists a subgraph $H'_{v}$ such that $H'_{v} \subset H_{t}(v)$ and $\rho(H'_{v}) \geq \rho^*(v)(1-\varepsilon)$.
\end{lemma}
\begin{proof}
    Suppose for the sake of contradiction that this is not the case.  
Let $G, v$ be a graph and a vertex contradicting the statement. 
Denote by $\overrightarrow{G}$ a locally fair orientation of $G$.
We define $V_i$ as the set of vertices to which $v$ has a directed path of length at most $i$ in $\overrightarrow{G}$. 
Since $\overrightarrow{G}$ is locally fair, all vertices $u \in V_i$ have an out-degree $\gee(x)$ of at least $\rho^*(v)$.

\subparagraph{Induction.}
 We show by induction that if there exists an integer $j$ where for all $i \leq j$ the density $\rho(G[V_i]) < (1-\varepsilon)\rho^*(v)$  then $|V_j| \geq (1-\varepsilon)^{-j}$. 

  Indeed, for the base case observe that $|V_0| = 1 \geq (1-\varepsilon)^{0}$. Now assume $|V_{j-1}| \geq (1-\varepsilon)^{-(j-1)}$. We show that it also holds for $j$. 

    For brevity, we denote for any set of edges $E'$ by $\omega(E') := \sum_{e \in E'} \gee(e)$. 
   All vertices in $V_{j-1}$ have an out-degree of at least $\rho^*(v)$, and their edge points to a vertex in $V_{j-1}$ or $V_j$. Hence:
    \[
    \omega(E(H[V_{j-1}])) + \omega(E(V_{j-1}, V_{j}-V_{j-1})) \geq \rho^*(v)|V_{j-1}|.
    \]
    
    This means that we have 
    \[
    (1-\varepsilon)\rho^*(v) \geq \rho(H(V_j)) \geq \frac{\omega(E(H[V_{j-1}])) + \omega(E(V_{j-1}, V_{j}-V_{j-1}))}{|V_{j}|} \geq \frac{\rho^*(v)|V_{j-1}|}{|V_{j}|}
    \]
    By applying the induction hypothesis, we find that indeed $|V_{j}| \geq (1-\varepsilon)^{-j}$. We conclude that if $\varepsilon \leq 1$ then:
    \begin{align*}
    j \leq \frac{\log n}{\log(\frac{1}{1-\varepsilon})} \leq \frac{\log n}{\log(1+\frac{\varepsilon}{1-\varepsilon})} \leq \frac{\log n}{\log(1+\varepsilon)} \leq \frac{\log n}{\log(1+\varepsilon)} \leq  \frac{2\log n}{\varepsilon}
    \end{align*}
    where the last inequality comes from the fact that $\log(1+x) \geq x-\frac{x^2}{2}$ when $x \geq 0$. 
    
    In particular this means that within the $t$-hop neighbourhood of $v$ in $H$ (for $t \in O(\eps^{-2} \log^2 n)$) we find a subgraph of density at least $ (1-\varepsilon)\rho^*(v)$ and density at most $\rho(H) = \rho^{*}(v)$. Since $H$ is a (weighted) subgraph of $G$, the same holds for $G$. 
\end{proof}

The subgraph reporting algorithm now works as follows: We first compute an $\eta$-fair fractional out-orientation for $\eta = \frac{\eps^2}{128\log n}$. 
From each selected vertex $v$; we run the standard CONGEST BFS-leader election algorithm (for example the one described in chapter 5 of~\cite{JukkaBook}) before truncating it after exactly $\lceil \frac{32 \log n}{\eps} \rceil$ rounds.
Instead of breaking ties on vertex IDs, we break ties on the fractional out-degree of vertices under the $\eta$-fair orientation, giving priority to the vertex with the largest out-degree (we then subsequently break ties by ID). 
Now a vertex is only elected a leader, if every vertex of its $\lceil \frac{32 \log n}{\eps} \rceil$-hop neighbourhood acknowledges it as the current leader once the BFS-leader election algorithm is truncated.
After truncation we spend $\Theta(\frac{\log n}{\eps})$ rounds informing the elected leaders of their status by sending acknowledgements from the $\lceil \frac{32 \log n}{\eps} \rceil$-hop neighbourhoods. 
Once a leader acknowledges that it is a leader, it broadcasts this information to its entire $\lceil \frac{32 \log n}{\eps} \rceil$-hop neighbourhood.
Note that it might be that some parts of the graph become leaderless -- i.e. they did not sit in the $\lceil \frac{32 \log n}{\eps} \rceil$-hop neighbourhood of an elected leader. 
If this is the case, these vertices never receive information that the leader election succeeded, and once sufficient time has passed, they simply output $0$. 

If a leader has fractional out-degree under the $\eta$-fair orientation that is smaller than $\tilde{D}$, the leader will ask its $\lceil \frac{32 \log n}{\eps} \rceil$-hop neighbourhood to all output $0$. 
If this is not the case, we say the leader is $\emph{active}$.
Note that we later show that at least one elected leader stays active if $\tilde{D} \leq \rho(G) $.

For each leader $\ell$, we then compute an $\eta$-fair out-orientation of the $\lceil \frac{32 \log n}{\eps} \rceil$-hop neighbourhood of the leader. 
The leader then uses the truncated BFS tree-structure to gather the value of the largest fractional out-degree $h_{\max}(\ell)$ under this newly computed $\eta$-fair out-orientation as well as the sizes of the sets $T_i(\ell) = \{u \in N^{t}(\ell): h(u) \geq h_{\max}(1+\eta)^{-i}\}$ for $i = 0$ up to $i = O(\frac{\log n}{\eps})$. 
Based on these numbers the leader can determine a cut-off point. That is: the leader calculates the smallest $k$ for which $|T_{k+1}| < (1+\frac{\eps}{16})|T_{k}|$. 
It then sets the cut-off point to be $h_{\max}(1+\eta)^{-(k+1)}$, and broadcasts the cut-off point to its $\lceil \frac{32 \log n}{\eps} \rceil$-hop neighbourhood. Finally, every vertex with a fractional out-degree greater than or equal to the cut-off point declares a $1$ and all other vertices declare $0$.

\subparagraph{Analysis:} We now show correctness of the above algorithm i.e.\ that it correctly terminates with a sufficiently dense subgraph.

We use this property to show that if a leader is active, then 
the nodes that output $1$ in the $\lceil \frac{32 \log n}{\eps} \rceil$ neighbourhood of the leader form a subgraph with density at least $(1 - \eps) \rho(G)$:
\begin{lemma} \label{lma:boundDens}
    Suppose $\ell$ is an active leader. Then the vertices in the $\lceil \frac{32 \log n}{\eps} \rceil$-hop neighbourhood of $\ell$ that output $1$ induce a subgraph of density at least $(1 - \eps) \rho^{*}(\ell)$.
\end{lemma}
\begin{proof}
    It follows from Lemma~\ref{lem:local} that the $\lceil \frac{32 \log n}{\eps} \rceil$-hop neighbourhood of $\ell$ contains a subgraph with density at least $\rho^{*}(\ell)(1-\frac{\eps}{16})$. 

    This in turn means that the largest fractional out-degree $h_{\max}$ under the second $\eta$-fair out-orientation computed only on the $\lceil \frac{32 \log n}{\eps} \rceil$-hop neighbourhood of $\ell$ is at least $\rho^{*}(\ell)(1-\frac{\eps}{16})$. 
    Indeed, this neighbourhood has maximum subgraph density at least $\rho^{*}(\ell)(1-\frac{\eps}{16})$ and so the largest fractional out-degree has to be at least $\rho^{*}(\ell)(1-\frac{\eps}{16})$. 

    This means that the density of the subgraph induced by all vertices outputting $1$, call this subgraph $H(\ell)$, can be bounded as follows: 
    \begin{align*}
    \rho(H(\ell)) &\geq \frac{\sum_{v \in T_k} g(v)}{|T_{k+1}|} 
     \geq  \frac{|T_k| \cdot{} h_{\max} \cdot{}(1+\eta)^{-k} }{(1+\frac{\eps}{16})|T_{k}|}  
     \geq  h_{\max}\frac{1}{(1+\frac{5\eps}{8})}  \Rightarrow \\
     \rho(H(\ell)) &\geq  \rho^{*}(\ell)(1-\frac{\eps}{16})(1-\frac{5\eps}{8}) \geq
      \rho^{*}(\ell)(1-\eps)
    \end{align*}
    where we used that $k \leq \log_{1+\eps/16}(n)$ since $T_k$ never can be larger than $n$, and that for $0 \leq x \leq 1$ we have $ \log 1+x \geq \frac{x}{2}$.
\end{proof}
Now we are ready to prove correctness: 
\begin{lemma}\label{lma:correctness}
    Let $G$, $\eps > 0$ and $\tilde{D}$ be given.
    Then the graph induced by all vertices declaring $1$ is either empty or has density at least $\tilde{D}(1-\varepsilon)$.
    Furthermore, if $\tilde{D} \leq \rho^{\max}(G)$ then at least one vertex will declare $1$.
\end{lemma}
\begin{proof}
    Note first that the $\lceil \frac{32 \log n}{\eps} \rceil$-hop neighbourhood of every active leader is disjoint. 
    This follows from the fact that every vertex acknowledges exactly one leader, and so two vertices sharing a vertex in their $\lceil \frac{32 \log n}{\eps} \rceil$-hop neighbourhood cannot both be acknowledged as leaders.
    This in turn implies that the subgraph $H$ induced by all vertices outputting a 1 has density at least $\tilde{D}(1-\eps)$. 

    Indeed, let $\ell_1, \dots, \ell_{s}$ be the active leaders. Then by above the induced subgraphs belonging to each leader, $H(\ell_1), \dots, H(\ell_s)$, are all disjoint and so the density of the graph $H = \bigcup \limits_{j} H(\ell_j) = (\bigcup\limits_{j} V(H(\ell_j)), \bigcup\limits_{j} E(H(\ell_j)) )$ satisfies:
    \begin{align*}
        \rho(H)  &= \frac{\sum \limits_{j} |E(H(\ell_j))|}{\sum \limits_{j} |V(H(\ell_j))|} = \frac{\sum \limits_{j} \rho(H(\ell_j))|V(H(\ell_j))|}{\sum \limits_{j} |V(H(\ell_j))|} \\
        &\geq \min \limits_{j}\{\rho(H(\ell_j))\}\frac{\sum \limits_{j} |V(H(\ell_j))|}{\sum \limits_{j} |V(H(\ell_j))|} \geq \min \limits_{j}\{\rho(H(\ell_j))\} \geq \tilde{D}(1-\eps)
    \end{align*}
    where the last inequality follows from Lemma~\ref{lma:boundDens}.

    Finally to show the furthermore part, observe that some vertex $v$ will receive a fractional out-degree at least $\rho(G)$ in the first computed $\eta$-fair orientation (indeed the $\rho(G)$ is the minimum maximum fractional out-degree achievable) and so the vertex with the smallest ID among the vertices with the largest fractional out-degree under the first $\eta$-fair orientation will always become an active leader.
\end{proof}
Finally, we have the following guarantee on the round complexity of the algorithm: 
\begin{theorem}
\label{thm:moregeneral}
    Suppose there exists an algorithm that given a graph $G$ on $n$ vertices and parameters $\eps$ computes an $\eta$-fair orientation in $O(U(n,\eta))$ rounds. 
    Then there exists an algorithm that, given $\tilde{D} \leq \rho^{\max}(G)$ as input at every vertex, outputs a subgraph $H$ with $\rho(H) \geq \rho(1 - \eps)\tilde{D}$ 
    after $O(U(n,\frac{\eps^2}{128\log n})+\frac{\log^2 n}{\eps^2})$ rounds.  
\end{theorem}
\begin{proof}
    Observe that computing the $\frac{\eps^2}{128\log n}$-fair orientations takes at most $O(U(n,\frac{\eps^2}{128\log n}))$ rounds.
    All other steps can be performed in $O(\frac{\log^2 n}{\eps^2})$ rounds. 
    Indeed, gathering or broadcasting information from a leader to the leaders $\lceil \frac{32 \log n}{\eps} \rceil$-hop neighbourhood takes at most $O(\frac{\log n}{\eps})$ rounds per piece of information needed.
    The most information intensive step is to gather the sizes of the sets $T_i$ (for the at most $\log_{1+\eps/16} n = O(\frac{\log n}{\eps})$ values of $i$), which can be done in $O(\frac{\log n}{\eps})$ rounds per $i$ that is in $O(\frac{\log^2 n}{\eps^2})$ rounds in total. 

    Finally, correctness of the algorithm follows from Lemma~\ref{lma:correctness}.
\end{proof}
In the variant, where we require that the output graph $H$ has a density that $(1+\eps)$-approximates $\rho^{\max}(G)$, we simply broadcast the maximum fractional out-degree $g^{\max}$ as before in $O(d)$ rounds. Observe that $g^{\max} \geq \rho^{\max}(G)$.
Then for a vertex to become an active leader it must have out-degree at least $g^{\max}$. 
Then we proceed as before. The above arguments can then be adapted mutatis mutandis to show that now $\rho(H) \geq (1-\eps)\rho^{\max}(G)$. 

Notice furthermore that if we want to solve the variant, where an input vertex $v$ is also specified, and we now wish to report a subgraph $H_v$ in the $t$-hop neighbourhood of $v$, containing $v$, with density at least $(1-\eps)\rho^*(v)$, we can modify by the above approach as follows. Let $k$ be such that Theorem~\ref{thm:newlocal} holds.
Next compute an $\eta$-fair orientation of the $k$-hop neighbourhood of $v$ instead. 
Finally, choose $v$ as the sole active leader, and let it report a dense subgraph in its $k$-hop neighbourhood following a similar protocol to before. The only difference is that now we define the the sets $T_{i}(v):= \{u\in N^{k}(v): h(u) \geq h(v)(1+\eta)^{-i}\}$. Other than that the protocol proceeds as before, but now $v \in T_{0}$, and so it is guaranteed to output $1$.
Making the change to the $T_{i}$'s mutatis mutandis to the arguments from Lemma~\ref{lma:boundDens} and Lemma~\ref{lem:local} show that the density of this subgraph will be at least $\rho^*(v)(1-\eps)$ as desired. 

By plugging in our algorithm for computing $\eta$-fair orientations, we arrive at the following result similarly to before: 
\begin{theorem}
\label{thm:vgeneral}
    Suppose there exists an algorithm that given a graph $G$ on $n$ vertices and parameters $\eps$ computes an $\eta$-fair orientation in $O(U(n,\eta))$ rounds. 
    Then there exists an algorithm that, given $v$ as input, outputs a subgraph $H$ with $\rho(H) \geq \rho(1 - \eps)\tilde{D}$ and $v\in H$
    after $O(U(n,\frac{\eps^2}{128\log n})+\frac{\log^3 n}{\eps^3})$ rounds.  
\end{theorem}
Note that the running time is a bit slower, as the height of the broadcast tree has increased from $O(\tfrac{\log n}{\eps})$ to $O(\tfrac{\log^2 n}{\eps^2})$.

\end{document}